\documentclass[journal]{IEEEtran}
\usepackage{amsmath}
\usepackage[T1]{fontenc}
\usepackage[latin9]{luainputenc}
\usepackage{amsbsy}
\usepackage{amstext}
\usepackage{amssymb}
\usepackage{graphicx}
\usepackage{esint}
\usepackage{amsthm}
\usepackage[unicode=true,
 bookmarks=true,bookmarksnumbered=true,bookmarksopen=true,bookmarksopenlevel=1,
 breaklinks=false,pdfborder={0 0 0},pdfborderstyle={},backref=false,colorlinks=false]
 {hyperref}
\newcommand{\tabincell}[2]{\begin{tabular}{@{}#1@{}}#2\end{tabular}}
\begin{document}
\title{Flow Field Reconstructions with GANs based on Radial Basis Functions}

\author{Liwei~Hu,~\IEEEmembership{Student Member,~IEEE,}
        Wenyong~Wang,~\IEEEmembership{Member,~IEEE,}
        Yu~Xiang,~\IEEEmembership{Member,~IEEE},
        Jun~Zhang
\thanks{W. Wang, professor, is with School of Computer Science and Engineering, University of Electronic Science and Technology of China, Chengdu 611731, China, e-mail:wangwy@uestc.edu.cn.}
\thanks{W. Wang, special-term professor, is with Macau University of Science and Technology, Macau 519020, China.}
\thanks{L. Hu, Y. Xiang, and J. Zhang are with University of Electronic Science and Technology of China, Chengdu 611731, China.}
\thanks{Manuscript received April XX, XXXX; revised August XX, XXXX.}}

\markboth{Journal of \LaTeX\ Class Files,~Vol.~14, No.~8, August~2015}%
{Shell \MakeLowercase{\textit{et al.}}: Bare Demo of IEEEtran.cls for IEEE Journals}

\maketitle

\begin{abstract}
Nonlinear sparse data regression and generation have been a long-term challenge, to cite the flow field reconstruction as a typical example. The huge computational cost of computational fluid dynamics (CFD) makes it much expensive for large scale CFD data producing, which is the reason why we need some cheaper ways to do this, of which the traditional reduced order models (ROMs) were promising but they couldn't generate a large number of full domain flow field data (FFD) to realize high-precision flow field reconstructions. Motivated by the problems of existing approaches and inspired by the success of the generative adversarial networks (GANs) in the field of computer vision, we prove an optimal discriminator theorem that the optimal discriminator of a GAN is a radial basis function neural network (RBFNN) while dealing with nonlinear sparse FFD regression and generation. Based on this theorem, two radial basis function- based GANs (RBF-GAN and RBFC-GAN), for regression and generation purposes, are proposed. Three different datasets are applied to verify the feasibility of our models. The results show that the performance of the RBF-GAN and the RBFC-GAN are better than that of GANs/cGANs by means of both the mean square error (MSE) and the mean square percentage error (MSPE). Besides, compared with GANs/cGANs, the stability of the RBF-GAN and the RBFC-GAN improve by 34.62\% and 72.31\%, respectively. Consequently, our proposed models can be used to generate full domain FFD from limited and sparse datasets, to meet the requirement of high-precision flow field reconstructions.
\end{abstract}


\IEEEpeerreviewmaketitle

\section{Introduction}
\IEEEPARstart{F}{low} field reconstructions \cite{yu2019flowfield} reflect the variation of air flows in spatial or time domain, which plays a fundamental role in the field of shape optimizations, inverse designs, and simulations of flight \cite{liu2017efficient,duan2020use,wang20203d}, etc. The term reconstructions denotes both the regression and generation of full domain flow field data (FFD) from given datasets \cite{sun2020physics}. Generally, considering the computational costs or inability of existing approaches of data producing, the number of data contained in flow field datasets is limited \cite{zhang2018application}. However the range of design parameters (e.g. height, Mach number, angle of attack, Reynolds number, etc) can be large, which results in the sparse nature of flow field datasets. Besides, the mapping between design parameters and response parameters (e.g. forces and moments, etc) is highly nonlinear. The sparse and nonlinear characteristics of flow field datasets make it difficult to reconstruct flow field precisely \cite{lu2018sparse}. In the area of aerodynamics, there are two different approaches to reconstruct flow field, namely full order models (FOMs) and reduced order models (ROMs) \cite{kou2017multi}.

The computational fluid dynamics (CFD), a representative of the FOMs, is a high-fidelity numerical simulation \cite{kumar2019computational}. CFD-based systems are the main approach to calculate high-fidelity FFD in aerodynamics. However, most of governing equations are nonlinear partial differential equations (PDEs) and have no numerical solutions. Therefore, it is impossible to calculate the aerodynamic response values under all flow states and aerodynamic shapes \cite{hu2020neural}. In addition, the large computational costs of CFD-based systems are beyond the tolerance of researchers both in time and hardware resources.

The drawbacks of FOMs motivated the development of ROMs, which can learn the nonlinear mapping between aerodynamic design parameters and response parameters from limited training data calculated by full order CFD-based systems \cite{white2020fast,su2016reduced}. Compared with FOMs, ROMs use different convergence strategies or models to reduce the calculation time on the premise of ensuring high accuracies. Besides, ROMs can deal with both linear and nonlinear FFD. As for linear FFD processing, the linear ROMs are considered, such as the autoregressive with exogenous input model \cite{su2016reduced}. As for nonlinear FFD processing, the nonlinear ROMs must be employed, such as Kriging model \cite{glaz2010reduced,da2011generation}, support vector machine \cite{chen2012support}, neural network \cite{mannarino2014nonlinear,ignatyev2015neural}, and block-oriented Wiener models \cite{kou2016novel}, etc. However, while dealing with nonlinear FFD, a large number of training data are needed to provided by CFD-based systems to guarantee the accuracies of ROMs. What's worse, ROMs can only output response values in given flow states, and cannot generate full domain FFD that satisfy the governing equations.

Neural networks, a part of the ROMs, have attracted much attention in aerodynamics in the past few years \cite{chen2020multiple,hu2020neural}. The conventional approaches are based on full connected networks (FCNs), which output response parameters under given flow states \cite{yu2019flowfield,milano2002neural}. However, the flow states need to be specified by researchers in advance and can not be automatically generated by FCN-based models. Hence these studies belongs to regression rather than generation. The generative adversarial networks (GANs), as new generative models, can generate full domain flow FFD from a random noise, which becomes a feasible way to solve this problem.

A GAN consists of two adversarial sub-models: the generator G and the discriminator D. The goal of D is to train a function $D(x)$ that can be used to distinguish the real data distribution $\mathbb{P}_{r}(x)$ from the generated data distribution $\mathbb{P}_{g}(x)$. The goal of G is to train a function $G(z)$ that can generate pseudo-data from a simple data distribution $\mathbb{P}_{z}(z)$. D and G are trained alternately until D can not distinguish the real data from the generated data. GANs have achieved great success in the field of image generation \cite{volkhonskiy2020steganographic,tasar2020colormapgan}, image translation \cite{lira2020ganhopper}, feature analysis \cite{zhang2020copy}, and machine translation \cite{wu2018adversarial}, etc. These studies demonstrated that GANs possess the capability to capture the distribution of data contained in datasets, and further map input into output space, which brings new ideas to the field of flow field reconstructions. Recently, researchers adopted GANs to reconstruct flow fields, which have become a state-of-the-art approach \cite{siddani2020machine}.

The existing studies in flow field reconstructions are categorized into image-based and non-image-based flow field reconstructions. In the field of image-based flow field reconstructions, the recursive CNNs-based GAN in unsteady flow predictions \cite{lee2019data}, the recurrent neural networks (RNNs)-based GAN in the analyses of the temporal continuity \cite{kim2020deep} and the tempoGAN for high-resolution flow field image generations \cite{xie2018tempogan} are outstanding models. All of them use flow field images as the input to train convolutional neural networks (CNN)-based GANs, and the outputs are still images. Although, these CNN-based GANs can generate high-resolution flow field images to express the variations of air flows, the non-image-based flow field reconstructions can omit the process of converting FFD into images while maintaining high accuracies. Farimani used the original FFD as a training set and adopted cGAN to study flow field reconstructions \cite{farimani2017deep}. The results showed that the MAE of generated FFD was lower than that of traditional ROMs. However, this non-image-based model is not stable enough (the error fluctuates rapidly while training).

Motivated by the problems of existing approaches and inspired by the adversarial learning mechanism of GANs, we prove an theorem that a radial basis function neural network (RBFNN) is the optimal discriminator of a GAN while dealing with nonlinear sparse data come from a continuous function. Based on this theorem, a radial basis function-based GAN (RBF-GAN) and a radial basis function cluster-based GAN (RBFC-GAN), two novel data-driven models, are proposed to implement both the regression and generation for nonlinear sparse FFD\footnote{The source code of RBF-GAN and RBFC-GAN can be found in https://github.com/huliwei123/RBF-GANs.git.}. Different from the GANs-based models mentioned above, the advantages of our models are a) they are both regression models and generative models, which can quickly generate full domain FFD in accordance with the variation lows in flow fields; b) they are data-based interpolation models, which do not require thorough knowledge of aerodynamics and can significantly reduce the error of generated data and improve the stability of GANs; and c) the training sets of the RBF-GAN and the RBFC-GAN are nonlinear sparse FFD, not images generated from FFD. Subsequently, three different datasets are applied to verify the validity and feasibility of the proposed models. The results show that a) the mean square error (MSE) and the mean square percentage error (MSPE) \cite{kou2017multi} of the data generated by the RBF-GAN/RBFC-GAN are lower than that of data generated by full connected neural networks-based GANs; and b) compared with already proposed cGAN, the stability of the RBF-GAN and the RBFC-GAN improves by 34.62\% and 72.31\%, respectively.

To summarize, the contributions of our work are:
\begin{enumerate}
\item[a)] we prove an theorem that an RBFNN is the optimal discriminator of a GAN while dealing with nonlinear sparse data;
\item[b)] based on a), we propose the RBF-GAN, both a generative and a regression model, to generate full domain FFD from limited nonlinear sparse data;
\item[c)] based on b), we propose the RBFC-GAN to further improve the fidelity of generated nonlinear sparse full domain FFD, as well as the stability of GANs after converged.
\end{enumerate}

The structure of this paper is as follows. Section II introduces the research status of RBFNNs and GANs in the field of aerodynamic response predictions and flow field reconstructions, respectively. In Section III, the optimal discriminator theorem for nonlinear sparse data generation and the mechanisms of the proposed RBF-GAN and RBFC-GAN are elaborated. In Section IV, three datasets from different scenes are applied to validate the effectiveness of the proposed models. The conclusion of our work is shown in Section V.

\section{Related Work}
In this section, we introduce the fundamental of RBFNNs and GANs, and their applications in FFD processing.

\subsection{RBFNNs}

\subsubsection{The Fundamental of RBFNNs}
An RBFNN (Fig.\ref{fig_RBFNN}), a regression model, is a three-layer feedforward neural network, which consists of one input layer, one hidden layer and one output layer \cite{moody1989fast}. Considering the combination of RBFNNs and GANs, the number of output neurons in this figure is set to 1. Different from fully connected networks (FCNs) which use sigmoid or ReLU as the activation function, an RBFNN employ an RBF instead \cite{tenney2018deep}. Generally, the most commonly used RBF is the Gaussian kernel function:

\begin{equation}
\nonumber
\label{equ_guassian_kernel_function}
\ensuremath{g_{1}\left(\mathbf{x}_{i},\mathbf{v}_{j},\sigma_{j}\right)=\exp\left(-\frac{\left\Vert \mathbf{x}_{i}-\mathbf{v}_{j}\right\Vert ^{2}}{2\sigma_{j}^{2}}\right)}
\end{equation}
where $||\cdot||$ is a norm on the function space, $\mathbf{x}_{i}$ denotes the $i$th sample in the training set, $\mathbf{v}_{j}$ and $\sigma_{j}$ denote the center and width of the $j$th hidden neuron, respectively. The output of an RBFNN can be written as:

\begin{equation}
\nonumber
\label{equ_rbfnn}
\ensuremath{\mathbf{y}_{i}=w_{0}+\sum_{j=1}^{q}\mathbf{w}_{j}\cdot g_{1}\left(\mathbf{x}_{i},\mathbf{v}_{j},\sigma_{j}\right)}
\end{equation}
where $q$ denotes the number of hidden neurons, $\mathbf{w}_{j}$ denotes the weight between one hidden neuron and one output neuron, and $w_{0}$ denotes the bias of the only output neurons.

\begin{figure}[tbh]
\centering
\includegraphics[scale=0.5]{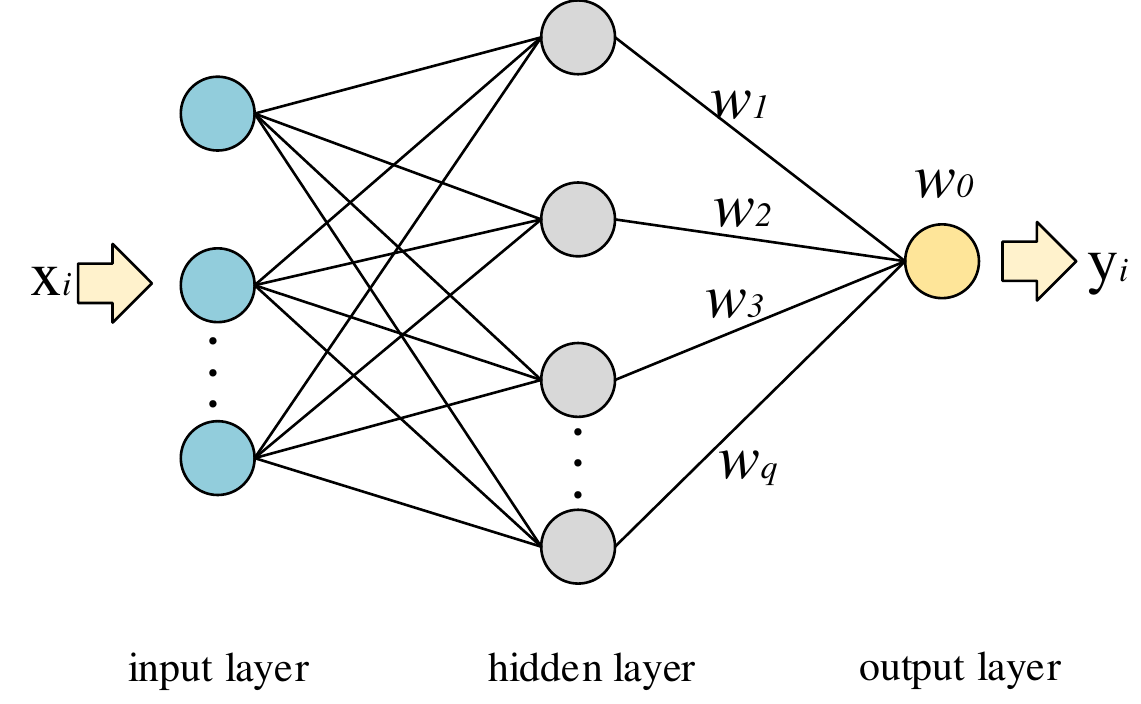}
\caption{The structure of an RBFNN. In particular, the center $\mathbf{v}_{j}$
of the RBF is a vector with the same dimension as the input $\mathbf{x}_{i}$,
while the width $\sigma_{j}$ is a scalar.}\label{fig_RBFNN}
\centering
\end{figure}

Theoretically, an RBFNN converges faster than a FCN with the same number of hidden neurons. Consider an input $\mathbf{x}$ and two hidden neurons $n_{1}$ and $n_{2}$ with parameters $\mathbf{v}_{1} \gg \mathbf{v}_{2}=\mathbf{x}$ and $\sigma_{1}=\sigma_{2}=1$, the output of $n_{1}$ is close to $\mathbf{0}$, while the output of $n_{2}$ is $\mathbf{1}$, according to (\ref{equ_guassian_kernel_function}). Therefore, $n_{1}$ produces a gradient of 0, which leads to the result that $\mathbf{v}_{1}$ and $\sigma_{1}$ remain unchanged, i.e. vanishing gradient \cite{tan2019vanishing}. In contrast, only $\mathbf{v}_{2}$ and $\sigma_{2}$ of neuron $n_{2}$ can be updated. The approach of updating parts of the parameters is called local updating mechanism, which perfectly explains why an RBFNN converges faster than a FCN with the same structure.

\subsubsection{RBFNN-based Aerodynamic ROMs }

RBFNNs-based aerodynamic ROMs are widely used to predict aerodynamic response parameters, because of its powerful nonlinear input-output mapping ability. According to the number of kernel functions (i.e. radial basis functions), the existing studies can be categorized into two groups: single kernel RBFNNs and multi-kernel RBFNNs.

Single kernel RBFNNs are the most commonly used models in aerodynamic response predictions. In 2005, Hovcevar et al. first adopted RBFNN with Gaussian kernel function to predict the fluctuations in the passive-tracer concentration for the turbulent wake behind an airfoil  \cite{hovcevar2005turbulent}. The explosion of RBFNNs in aerodynamic data modeling came after 2010 \cite{duraisamy2017status}. Ghoreyshi applied single kernel RBFNNs to study unsteady aerodynamic modeling at low-speed flow conditions \cite{ghoreyshi2013computational}. Different from Ghoreyshi, Zhang divided the whole flow field into three parts: the near-wall region, the wake region and the far-field region \cite{zhang2018machine,zhang2010high}. Three single kernel RBFNNs were built to predict the eddy velocity. All these models mentioned above use the Gaussian kernel function.

Multi-kernel RBFNNs refer to those RBFNNs with multiple RBFs, which are recently studied. Kou et al. proposed a multi-kernel RBFNN, and used it to model unsteady aerodynamics \cite{kou2017multi}. Kou implemented this model by adopting the linear combination of the Gaussian kernel function and wavelet kernel function mentioned in \cite{zhang2004wavelet}. The experimental results showed that the errors of multi-kernel neural networks is lower than that of single kernel RBFNNs in both fixed Mach number test case and variant Mach number test case.

The RBFNNs mentioned above are not applied in flow field reconstructions, because RBFNNs are interpolation models and do not have the ability to generate abundant full domain FFD. Hence, it is difficult for RBFNNs to achieve rapid and refined flow field management.

\begin{figure*}[tbh]
\centering
\includegraphics[scale=0.7]{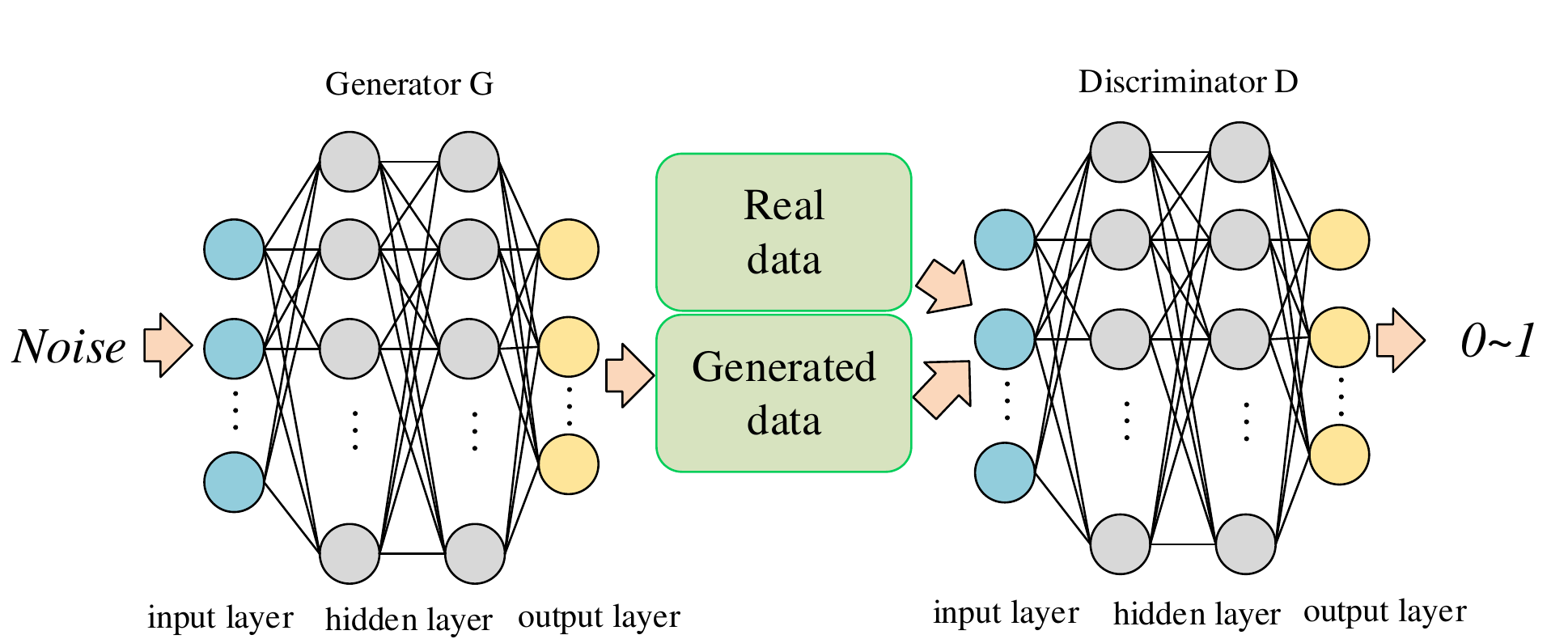}
\caption{The structure of a classic GAN.}\label{fig_GAN}
\centering
\end{figure*}

\subsection{GANs}

\subsubsection{The Fundamental of GANs}

Fig.\ref{fig_GAN} illustrates the structure of a GAN, which takes a noise with specific distributions as input and outputs a variable with a value between 0 and 1. The loss functions of D and G are:

\begin{equation}\label{equ_loss_D}
L_{D}=-\mathbb{E}_{\boldsymbol{x}\sim p_{\text{r }}(\boldsymbol{x})}[\log D(\boldsymbol{x})]-\mathbb{E}_{\boldsymbol{z}\sim p_{\boldsymbol{z}}(z)}[\log(1-D(G(\boldsymbol{z})))]
\end{equation}

\begin{equation}\label{equ_loss_G}
\ensuremath{L_{G}=\mathbb{E}_{z\sim p_{z}(z)}[\log(1-D(G(z)))]}
\end{equation}
where $\boldsymbol{x}$ denotes the real data that follow the distribution density $p_{r}(x)$, $\boldsymbol{z}$ denotes the input noise with the distribution density $p_{z}(z)$. The training goal of a GAN is to make the above two formulas reach their minimum.

Although GANs have achieved great success in many fields, they still face many problems, for example mode collapse \cite{mordido2020microbatchgan,li2020efficient}, convergence problem \cite{arjovsky2017towards}, vanishing gradient \cite{arjovsky2017wasserstein}, application problems, instability \cite{hu2020neural}, etc. To solve these problems, some new models, algorithms and loss functions were proposed \cite{creswell2018generative}. The Wasserstein-GAN and earth-mover distance \cite{arjovsky2017wasserstein} essentially solved the problem of convergence, and they alleviate the problem of mode collapse. The Dual-GAN \cite{zhao2017dual} and the DC-GAN \cite{li2018unsupervised} expands the application of GANs in the field of computer vision \cite{yang2018learning}. The conditional GAN (cGAN) \cite{ye2020deep} and the info-GAN \cite{chen2016infogan} can generate desired data under specific conditions, which expand the basic application of GANs. Based on current studies, the problems of model collapse and instability are still the major issues of GANs. In this paper, our proposed models alleviate the model collapse and instability, which are discussed in Section IV.B.1) and IV.B.3), respectively.

\subsubsection{Generative Adversarial Networks-based ROMs for nonlinear sparse FFD regression and generation}

In this section, we only discuss the non-image-based studies. Usually, the cGAN rather than GAN is widely used. Because most of the governing equations (e.g.  Navier-Stokes equation) are partial differential equations and have no numerical solutions, which makes it infeasible to use the governing equations to evaluate errors of GANs-based models.

Farimani adopted cGANs, both regression models and generative models, to learn the transport phenomena in aerodynamics \cite{farimani2017deep}. The results show that the MAE of generated data is lower than 0.01, however, the stability of cGAN-based models is not discussed. Hu et al. adopted cGANs to predict the velocity $u$ in Burgers' equation and the coefficient of pressure $C_{p}$ in Navier-Stokes equation \cite{hu2020neural}. By comparing cGANs with different structures, the stability of them is emphasized: after converged, the errors of cGANs during the training procedure fluctuates rapidly in a relatively large range, which result in that a) the designer of the models cannot decide when to stop the training process according to their experiences; and b) the models cannot generate high-fidelity full domain FFD.

As a result, existing non-image-based models are not suitable for the case of using nonlinear sparse FFD as data sets, because that the errors of generated FFD is relatively large, or they are not stable. To overcome the drawbacks of existing models in flow field reconstructions and considering the characteristics of FFD, we prove that an RBFNN is the optimal discriminator of a GAN in next section. Based on this proof, the RBFNNs are introduced into GANs, which naturally form two novel GANs named RBF-GAN and the RBFC-GAN. Compared with existing approaches \cite{pan2020tec,zhao2017dual,wang2018sentigan}, these two models are more accurate and stable.

\section{Methodology}

\subsection{Optimal Discriminator of GANs for Nonlinear Sparse Data generation}

GoodFellow et al. proved that the optimal discriminator exists where $P_{r}(x)=P_{g}(x)$, i.e. $D(x)=0.5$ \cite{goodfellow2014generative}. However, the form of the optimal discriminator of a GAN was still unsolved. In this section, we first introduce the best approximation theory to clarify the nature of the neural networks to approximate given continuous functions, and then we prove that an RBFNN is the optimal discriminator while dealing with nonlinear sparse data come from these functions.

\subsubsection{The Best Approximation Problem}

The best approximation problem of neural networks could be stated as follows:

\newtheorem{definition}{\bf Definition}[section]
\begin{definition}
\label{definition_of_optimal_appeoximation}
Given a continuous function $F(X)$ belongs to function space $\Phi$ and a subset $A \subseteq \Phi$, find a function $f(W,X) \epsilon A$ that depends continuously on $W \epsilon R^{n}$ and $X$, such that

\begin{equation}
\label{equ_approximation}
\ensuremath{d(F(X),A)\equiv\inf_{f\in A}\|F(X)-f(W,X)\|}
\end{equation}
for all $W$ in the set $R^{n}$, where $d(F(X),A)$ denotes the distance of $F(X)$ from subset $A$. If the infimun of $||F(X)-f(W,X)||$ is attained for some function $f(W^{*},X)$ of $A$, we suppose that $f(W^{*},X)$ is the best approximation for $F(X)$.
\end{definition}

The approximation problem of neural networks is to construct a smooth mapping $f(W,X)$ between inputs and outputs from limited data. The smooth mapping denotes that small variations in inputs result in small variations in corresponding outputs. Paradoxically, the information contained in the nonlinear sparse datasets is not enough to reflect these small variations. For example, unlike tens of thousands of samples in computer vision, the sample number of a general aerodynamic dataset is usually less than 1000 and hard to increase \cite{zhang2018application,white2020fast,li2020efficient}. Besides, the high nonlinear characteristic of aerodynamic data undoubtedly increases the difficulty of approximations. Therefore, the best approximation is the only way to break this dilemma. Formula (\ref{equ_approximation}) denotes a function of the smooth mapping $f(W,X)$, whose minimum value can be calculated by Euler-Lagrange equation, which is introduced in next section.

It was proved that a FCN does not have the optimal approximation property for any continuous function, but an RBFNN does \cite{girosi1990networks}. Therefore, in next sections, we prove a theorem that an RBFNN is the optimal discriminator of a GAN while dealing with nonlinear sparse data.

\subsubsection{The Optimal Discriminator}

\newtheorem{thm}{\bf Theorem}[section]
\begin{thm}\label{theorem_best_D}
Assuming that $\mathcal{D}=\{(z_{i},x_{i})\epsilon R^{n}|i=1,\cdots N\}$ is a nonlinear and sparse dataset, where $x_{i}$ denotes the real data, and $z_{i}$ denotes the noise with a particular distribution, then the optimal discriminator of a GAN is an RBFNN.
\end{thm}
\begin{proof}\renewcommand{\qedsymbol}{}
To prove Theorem \ref{theorem_best_D}, we introduce the regularization theory \cite{girosi1990networks}:

\begin{equation}\label{equ_regularization}
\ensuremath{\sum_{i}\left(f(x_{i})-d_{i}\right)^{2}+\lambda\left\Vert Pf(x)\right\Vert ^{2}}
\end{equation}
where $\sum_{i}\left(f(x_{i})-d_{i}\right)^{2}$ is the original loss function of a neural network, which measures the distance between the generated data $f(x_{i})$ and the real data $d_{i}$, $\left\Vert Pf(x)\right\Vert ^{2}$ embeds a priori information on the function $f(x)$, $P$ is a differential constrain operator \cite{poggio1990networks}, and $\lambda$ is the coefficient of the regular term.

Based on (\ref{equ_regularization}), we write (\ref{equ_loss_D}) as following:

\begin{equation}
\label{equ_loss_D_norm}
\begin{aligned}
L_{D}=&-\mathbb{E}_{\boldsymbol{x}\sim p_{\text{r }}(\boldsymbol{x})}[\log D(\boldsymbol{x})]-\\
&\mathbb{E}_{\boldsymbol{z}\sim p_{\boldsymbol{z}}(z)}[\log(1-D(G(\boldsymbol{z})))]+\lambda||PD(x)||^{2}
\end{aligned}
\end{equation}

Formula (\ref{equ_loss_D_norm}) describes the loss function of $D(x)$, and can be written as the form of a functional:

\begin{equation}
\label{equ_loss_D_functional}
\begin{aligned}
H[D(x)]=&-\mathbb{E}_{\boldsymbol{x}\sim p_{\text{r }}(\boldsymbol{x})}[\log D(\boldsymbol{x})]-\\
&\mathbb{E}_{\boldsymbol{z}\sim p_{\boldsymbol{z}}(z)}[\log(1-D(G(\boldsymbol{z})))]+\lambda||PD(x)||^{2}
\end{aligned}
\end{equation}

As proved in \cite{goodfellow2014generative}, the loss function of a discriminator reached its minimum value while $D(x)=0.5$, which means that the minimum value of (\ref{equ_loss_D_functional}) does exist. Therefore, the problem of the optimal discriminator of a GAN is to determine the function $D(x)$ such that (\ref{equ_loss_D_functional}) reaches its minimum value, which we use Euler-Lagrange equation to calculate the minimum value of $H[D(x)]$ to accomplish the optimal approximation defined in Theorem \ref{definition_of_optimal_appeoximation}.

Formula (\ref{equ_loss_D_functional}) could be expanded into a discrete form:

\begin{equation}\label{equ_functional_sum}
\begin{aligned}
H[D(x)]=&-\sum_{i=1}^{N}\{\mathbb{P}_{r}(x_{i})logD(x_{i})\\
&+\mathbb{P}_{g}(x_{i})log[1-D(x_{i})]\}+\lambda||PD(x)||^{2}
\end{aligned}
\end{equation}

To minimize the functional $H[D(x)]$, we introduce the Euler-Lagrange equation \cite{courant1962}:

\begin{equation}\label{equ_euler_lagrange}
\frac{\partial L}{\partial f\left(x\right)}-\frac{d}{dx}(\frac{\partial L}{\partial f'\left(x\right)})\equiv0
\end{equation}
where $L$ is a functional of $x$, $f\left(x\right)$ and $f'\left(x\right)$. Similarly, $H[\cdot]$ in (\ref{equ_functional_sum}) is a functional of $x$, $D\left(x\right)$ and $D'\left(x\right)$, i.e. $H[x, D\left(x\right), D'\left(x\right)]$. However, there is no $D'\left(x\right)$ in (\ref{equ_functional_sum}), the functional $H[\cdot]$ can be written as:

\begin{equation}\label{equ_functional_sum_HXDX}
\begin{aligned}
H[x,D(x),D'(x)]&=H(x,D(x))\\&=-\sum_{i=1}^{N}\{\mathbb{P}_{r}(x_{i})logD(x_{i})\\
&+\mathbb{P}_{g}(x_{i})log[1-D(x_{i})]\}+\lambda||PD(x)||^{2}
\end{aligned}
\end{equation}

From (\ref{equ_functional_sum_HXDX}), we get:
\begin{small}
\begin{equation}
\label{equ_partial_derivative HD}
\begin{aligned}
\begin{cases}
\frac{\partial H}{\partial D(x)}=-\sum_{i=1}^{N}[\mathbb{P}_{r}(x_{i})\frac{1}{D(x_{i})}-\mathbb{P}_{g}(x_{i})\frac{1}{1-D(x_{i})}]+2\lambda\hat{P}PD(x)\\
\frac{d}{dx}(\frac{\partial H}{\partial D'\left(x\right)})\equiv0
\end{cases}
\end{aligned}
\end{equation}
\end{small}
where $\hat{P}$ is the adjoint of the differential constrained operator $P$.

Associated with (\ref{equ_euler_lagrange}) and (\ref{equ_partial_derivative HD}), we get:

\begin{equation}\label{equ_EL_export}
-\sum_{i=1}^{N}[\mathbb{P}_{r}(x_{i})\frac{1}{D(x_{i})}-\mathbb{P}_{g}(x_{i})\frac{1}{1-D(x_{i})}]+2\lambda\hat{P}PD(x)=0
\end{equation}

Then:

\begin{equation}\label{equ_export}
\hat{P}PD(x)=\frac{1}{2\lambda}\sum_{i=1}^{N}[\mathbb{P}_{r}(x_{i})\frac{1}{D(x_{i})}-\mathbb{P}_{g}(x_{i})\frac{1}{1-D(x_{i})}]
\end{equation}
where, $\sum_{i=1}^{N}[\mathbb{P}_{r}(x_{i})\frac{1}{D(x_{i})}-\mathbb{P}_{g}(x_{i})\frac{1}{1-D(x_{i})}]$ in the right part represents the distance between the real data distribution and the generated data distribution, which is consistent with the training mechanism of GNAs, and $\hat{P}PD(x)$ represents the optimal discriminator constrained by the differential constrained operator $P$.

\begin{figure*}[tbh]
\centering
\includegraphics[scale=0.7]{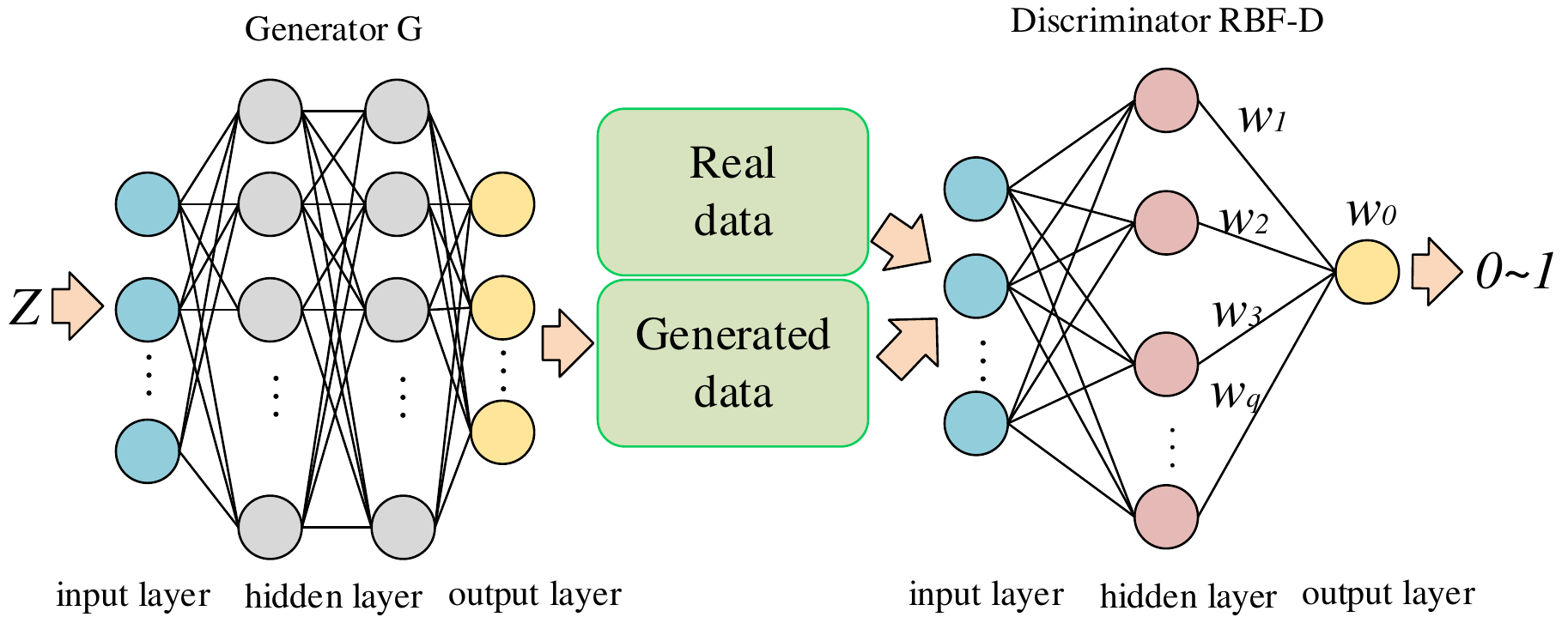}
\caption{The structure of an RBF-GAN. In the hidden layer of the discriminator
D, the pink neurons represent RBFs. In this paper, we chose the Gaussian
kernel function to realize the RBF-based discriminator D.}\label{fig_structure_of_RBF_GAN}
\centering
\end{figure*}

Since FFD in the training set is usually limited and sparse, it is hard to capture the small variations of inputs or outputs. Therefore, the information contained in every FFD should be fully utilized to achieve more accurate approximation than existing GANs. Based on (\ref{equ_export}), we introduce the Dirac function \cite{yerokhin2020calculations} that is widely used in the density problems in discrete distribution (or sparse data) in space or time to obtain :

\begin{equation}\label{equ_export_with_dirac}
\hat{P}PD(x)=\frac{1}{2\lambda}\sum_{i=1}^{N}[\mathbb{P}_{r}(x_{i})\frac{1}{D(x_{i})}-\mathbb{P}_{g}(x_{i})\frac{1}{1-D(x_{i})}]\delta(x-x_{i})
\end{equation}
where $\delta(x-x_{i})$ is the Dirac function, implies that $D(x)$ is only activated where $x=x_{i}$ precisely and de-activated where $x \neq x_{i}$. To simplify (\ref{equ_export_with_dirac}), we introduce Green's function $R(x;y)$, which satisfies the following distributional formula \cite{girosi1990networks}:

\begin{equation}
\label{equ_green}
\hat{P}PR(x;y)=\delta(x-y)
\end{equation}
The Dirac function $\delta(x-x_{i})$ in (\ref{equ_export_with_dirac}) is replaced with (\ref{equ_green}) to get the following formula:

\begin{equation}
\begin{aligned}
\label{equ_export_with_operator}
\hat{P}PD(x)&=\frac{1}{2\lambda}\sum_{i=1}^{N}[\mathbb{P}_{r}(x_{i})\frac{1}{D(x_{i})}\\&-\mathbb{P}_{g}(x_{i})\frac{1}{1-D(x_{i})}]\hat{P}PR(x;x_{i})
\end{aligned}
\end{equation}

From (\ref{equ_export_with_operator}), we derive the function $D(x)$:

\begin{equation}\label{equ_final_form}
D(x)=\frac{1}{2\lambda}\sum_{i=1}^{N}[\mathbb{P}_{r}(x_{i})\frac{1}{D(x_{i})}-\mathbb{P}_{g}(x_{i})\frac{1}{1-D(x_{i})}]R(x;x_{i})
\end{equation}
where $\mathbb{P}_{r}(x_{i})\frac{1}{D(x_{i})}-\mathbb{P}_{g}(x_{i})\frac{1}{1-D(x_{i})}$
is the distance of the real data from generated data. Let $w_{i}=\mathbb{P}_{r}(x_{i})\frac{1}{D(x_{i})}-\mathbb{P}_{g}(x_{i})\frac{1}{1-D(x_{i})}$ denotes the weights that the discriminator needs to learn, the function $D(x)$ can be written as in (\ref{equ_rbfnn_form}).

\begin{equation}\label{equ_rbfnn_form}
D(x)=\frac{1}{2\lambda}\sum_{i=1}^{N}w_{i}R(x;x_{i})
\end{equation}

In practical applications, the differential constrained operation $P$ is  rotationally and translationally invariant \cite{poggio1990networks}, therefore, $R$ in (\ref{equ_rbfnn_form}) will depend on the difference of its arguments, i.e. $R(x;x_{i})=R(||x-x_{i}||)$. Consequently, we get the final form:

\begin{equation}\label{equ_rbfnn_form_final}
D(x)=\frac{1}{2\lambda}\sum_{i=1}^{N}w_{i}R(||x-x_{i}||)
\end{equation}
where the function $R(\cdot)$ could be replace with any radial basis function, for example Gaussian kernel function, and introduce the width $\sigma$, we get an RBFNN based on Gaussian kernel functions:

\begin{equation}\label{equ_rbfnn_form_gaussian}
D(x)=\frac{1}{2\lambda}\sum_{i=1}^{N}w_{i}g_{1}(\mathbf{x}_{i},\mathbf{v},\sigma)
\end{equation}
\end{proof}

As a result, the loss function of a discriminator could be derived into a form of RBFNNs, which enables RBF-based discriminators to focus on the nonlinear and sparse data and ensure that the model has the minimum loss in the training set, that is, the hypersurface information contained in the training set is better utilized. Inspired by this theorem, RBFNNs can serve as the discriminator of a GAN while dealing with nonlinear spare data. Therefore, we introduce the GANs based on RBFNNs (i.e. RBF-GAN and RBFC-GAN) in next subsection.

\subsection{RBF-GAN and RBFC-GAN}

\subsubsection{Overview}
We proposed an RBF-GAN and an RBFC-GAN based on the Theorem \ref{theorem_best_D} for nonlinear sparse data regression and generation. The RBF-GAN or RBFC-GAN also consists of a G and a D. The G is a FCN with multiple hidden layers. While the D is an RBFNN with only one hidden layers. We call the D in an RBF-GAN as the RBF-D, and similarly, the D in an RBFC-GAN is called the RBFC-D. In this section, we discuss the RBF-GAN and RBFC-GAN, separately.

\subsubsection{RBF-GAN}
We introduce multiple Gaussian kernel functions into (\ref{equ_rbfnn_form_gaussian}) to form an novel RBF-D:

\begin{equation}\label{equ_rbf_D}
D(\mathbf{X})=w_{0}+\frac{1}{2\lambda}\sum_{i=1}^{N}\sum_{j=1}^{q}w_{ij} \cdot g_{1}(\mathbf{x}_{i},\mathbf{v}_{j},\sigma_{j})
\end{equation}
where $\mathbf{X}$ denotes the input matrix, $q$ denotes the number of Gaussian kernel functions, which equals to the number of hidden neurons in an RBF-D, $N$ denotes the number of data $\mathbf{X}$, $w_{ij}$ denotes the weight of the $j$th hidden neuron while dealing with the $i$th data, and $w_{0}$ denotes the bias of the only output neuron. We omit the coefficient $\frac{1}{2\lambda}$, and focus on a specific input vector $\mathbf{x}_{i}$, the (\ref{equ_rbf_D}) can be written as:

\begin{equation}\label{equ_output_of_rbf_GAN}
D(\mathbf{x}_{i})=w_{0}+\sum_{j=1}^{q}w_{j} \cdot g_{1}(\mathbf{x}_{i},\mathbf{v}_{j},\sigma_{j})
\end{equation}
where $w_{j}$ denotes the weight between one hidden neuron and the only output neuron. This formula indicates that an RBF-based discriminator consists of three different layers: one input layer (i.e. $\mathbf{x}_{i}$), one hidden layer (i.e. the Gaussian kernel function $\sum_{j=1}^{q}w_{j}g_{1}(\mathbf{x}_{i},\mathbf{v}_{j},\sigma_{j})$) and one output layer (i.e. $D(\mathbf{x}_{i})$). Fig.\ref{fig_structure_of_RBF_GAN} shows the structure of an RBF-GAN, where $Z$ is a noise with a distribution $\mathbb{P}_{z}$, ``Real data'' denotes the data with distribution $\mathbb{P}_{r}(x)$, while ``Generated data'' with distribution $\mathbb{P}_{g}(x)$ denotes the data generated by G. The output of the RBF-D is a variable with a value between 0 and 1, which indicates the probability of the generated data is realistic.

\begin{figure*}[tbh]
\centering
\includegraphics[scale=0.7]{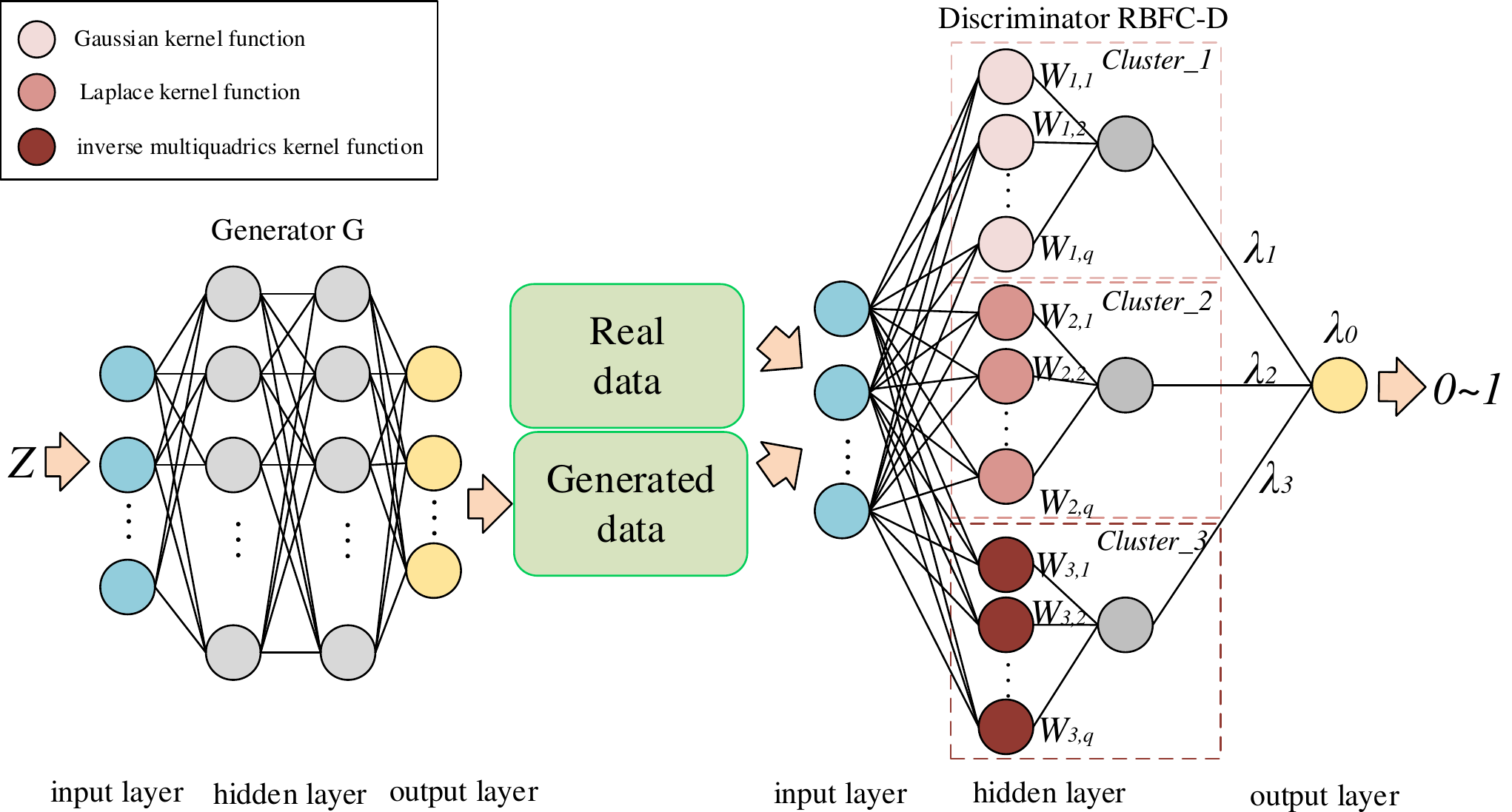}
\caption{The structure of an RBFC-GAN. Different from the RBF-D, the RBFC-D consists of multiple RBFNNs with different kernel functions. The RBFC-GAN is realized through combing different RBFNNs linearly.}\label{fig_structure_of_RBFC_GAN}
\centering
\end{figure*}

\subsubsection{RBFC-GAN}

We replace the Gaussian kernel functions in (\ref{equ_output_of_rbf_GAN}) with multiple RBFs, then the following formula can be obtained:

\begin{equation}\label{equ_output_of_RBFC_GAN}
D(\mathbf{x}_{i})=\lambda_{0}+\mathop{\sum_{z=1}^{m}}\lambda_{z}\sum_{j=1}^{q}w_{z,j}\cdot g_{z}\left(\mathbf{x}_{i},\mathbf{v}_{j},\sigma_{j}\right)
\end{equation}
where $m$ denotes the number of RBFs (i.e. the number of RBFNNs), $q$ denotes the number of hidden neurons in every RBFNN, $w_{z,j}$ denotes the weight of the $j$th hidden neuron in the $z$th RBFNN, and $\lambda_{z}$ denotes the coefficient for the $z$th RBFNN.

Formula (\ref{equ_output_of_RBFC_GAN}) describes the structure of an RBFC-D that consists of one input layer (i.e. $\mathbf{x}_{i}$), one hidden layer (i.e. $\sum_{z=1}^{m}\lambda_{z}\sum_{j=1}^{q}w_{z,j}\cdot g_{z}(x_{i},v_{j},\sigma_{j})$) and one output layer (i.e. $D(\mathbf{x}_{i})$). Different from the RBF-D, the hidden layer of an RBFC-D consists of multiple clusters of RBFNNs. In this paper we call a RBFNN (i.e. $\sum_{j=1}^{q}w_{z,j}\cdot g_{z}\left(\mathbf{x}_{i},\mathbf{v}_{j},\sigma_{j}\right)$) as a cluster. Therefore, Fig.\ref{fig_structure_of_RBFC_GAN} shows the structure conducted by (\ref{equ_output_of_RBFC_GAN}) with $m=3$, which means the RBFC-D consists of three different clusters (i.e. RBFNNs). The RBFs in every cluster are as follows:

\begin{equation}
\label{equ_multi_kernel_function}
\begin{cases}
\ensuremath{g_{1}\left(\mathbf{x}_{i},\mathbf{v}_{j,1},\sigma_{j,1}\right)=\exp\left(-\frac{\left\Vert \mathbf{x}_{i}-\mathbf{v}_{j,1}\right\Vert ^{2}}{2\sigma_{j,1}^{2}}\right)}\\
\ensuremath{g_{2}(\mathbf{x}_{i},\mathbf{v}_{j,2},\sigma_{j,2})=\exp(-\frac{||\mathbf{x}_{i}-\mathbf{v}_{j,2}||}{\sigma_{j,2}})}\\
\ensuremath{g_{3}(\mathbf{x}_{i},\mathbf{v}_{j,3},\sigma_{j,3})=\frac{1}{\sqrt{(\mathbf{x}_{i}-\mathbf{v}_{j,3})^{2}+\sigma_{j,3}^{2}}}}
\end{cases}
\end{equation}
where $g_{1}$, $g_{2}$ and $g_{3}$ denotes the Gaussian kernel function, the Laplace kernel function and the inverse multiquadrics kernel function, respectively, $\mathbf{v}_{j,z}$ and $\sigma_{j,z}$ with $z=1,2,3$ denote the center and width of $j$th RBF in $z$th RBFNN, respectively.

\section{Experiments And Results Analysis}

\subsection{DataSets}

To validate the effectiveness and feasibility of the RBF-GAN and the RBFC-GAN, three different datasets are employed, namely the Burgers' dataset, the cylindrical laminar dataset and the ONERA M6 dataset.

\subsubsection{Burgers' Dataset}
\begin{figure}[th]
\centering
\includegraphics[scale=0.3]{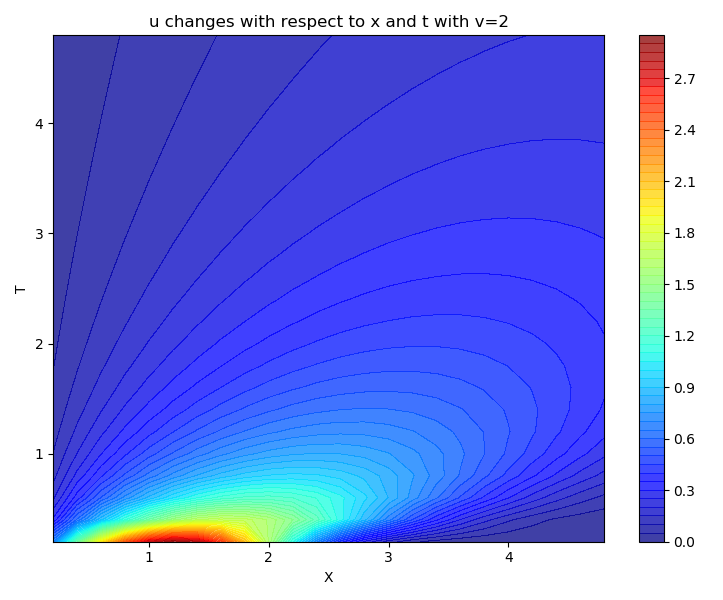}
\caption{The visualization results of the Burgers' Equation data. In this figure, the
x-coordinate represents the displacement $x$, while the y-coordinate
represents the time $t$.}\label{fig_visualization_of_burgers_dataset}
\centering
\end{figure}

The Burgers' equation is an one-dimensional PDE that expresses the movement of a shockwave across a tube:

\begin{equation}
\nonumber
\label{equ_burgers_equation}
\ensuremath{\frac{\partial u}{\partial t}+u\frac{\partial u}{\partial x}=v\frac{\partial^{2}u}{\partial x^{2}}}
\end{equation}
where $u$ denotes the velocity of the fluid (the shockwave), $t$ denotes the time, $x$ denotes the displacement, and $v$ denotes the viscosity coefficient. The variation range of design parameters ($t$, $x$ and $v$) is from 0.2 to 4.8, and the step is 0.2. Consequently, this dataset contains 13824 samples. Fig.\ref{fig_visualization_of_burgers_dataset} shows the variation of velocity $u$ with respect to time $t$ and displacement $x$ under the condition of $v=2.0$. In this figure, the red part in the lower left corner describes the variations of high-velocity flows. Therefore, this Burgers' dataset contains limited high-velocity nonlinear data that can be used to verify whether the model will cause mode collapses.

\subsubsection{Cylindrical Laminar Dataset}

The cylindrical laminar dataset is a two-dimensional application of the Navier-Stokes equation that simulates the pressure change on the surface of a cylinder while a flow passes the cylinder. The Navier-Stokes equation is as follows:

\begin{equation}
\nonumber
\label{equ_NS_equation}
\ensuremath{\rho\left[\frac{\partial V}{\partial t}+(V.\nabla)V\right]=-\nabla P+\rho g+v\nabla^{2}V}
\end{equation}
where $P$, $V$, $t$, $\rho$, $v$ denote the pressure, velocity, time, density and coefficient of viscosity, respectively. In this dataset, we used SU2 to calculate 6000 sample points \cite{vitale2020multistage}. The format and the variation range of the design parameters in this dataset are shown in Tab.\ref{tab_NS_input_data_format}. The format of the response parameters is shown in Tab.\ref{tab_NS_output_data_format}. Fig.\ref{fig_visualization_of_NS_dataset} shows the distribution of $C_{p}$ around the cylindrical with $v=2.0$.

\begin{table}[tbh]
\centering
\begin{tabular}{cccc}
\hline
variables & $x$ & $y$ & $Ma$\tabularnewline
\hline
significance & x-coordinate & y-coordinate & mach number\tabularnewline
start & 0.1 & -0.5 & 0.1\tabularnewline
end & 1 & 0.5 & 0.24\tabularnewline
step & 0.005 & 0.0078 & 0.01\tabularnewline
\hline
\end{tabular}
\caption{The data format and variation range of design parameters in cylindrical
laminar dataset.}\label{tab_NS_input_data_format}
\centering
\end{table}

\begin{table*}[tbh]
\centering
\begin{tabular}{ccccc}
\hline
variables & $P$ & $C_{p}$ & $F_{x}$ & $F_{y}$\tabularnewline
\hline
significance & pressure & coefficient of pressure & the friction in the x-coordinate & the friction in the x-coordinate\tabularnewline
\hline
\end{tabular}
\caption{The data format of response parameters in cylindrical laminar test
case.}\label{tab_NS_output_data_format}
\centering
\end{table*}

\begin{figure}[tbh]
\centering
\includegraphics[scale=0.3]{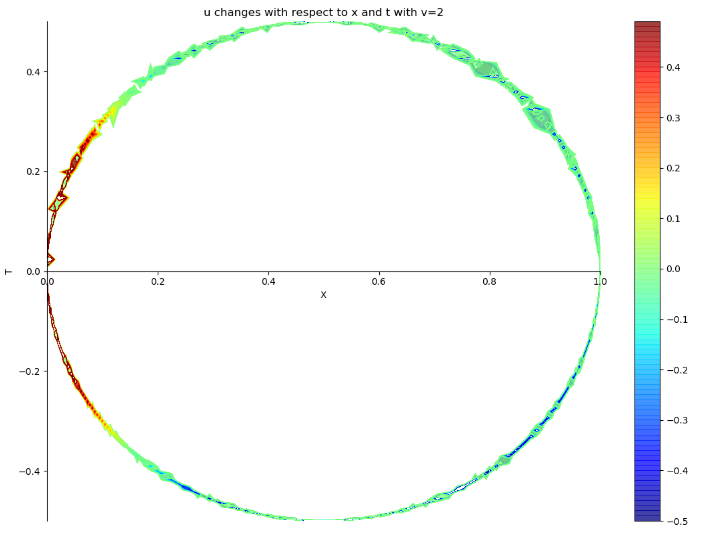}
\caption{The visualization results of the cylindrical laminar dataset. In this figure,
the x-coordinate represents the displacement $x$, while the y-coordinate
represents the time $t$.}\label{fig_visualization_of_NS_dataset}
\centering
\end{figure}

\begin{figure*}[bth]
\centering
\includegraphics[scale=0.25]{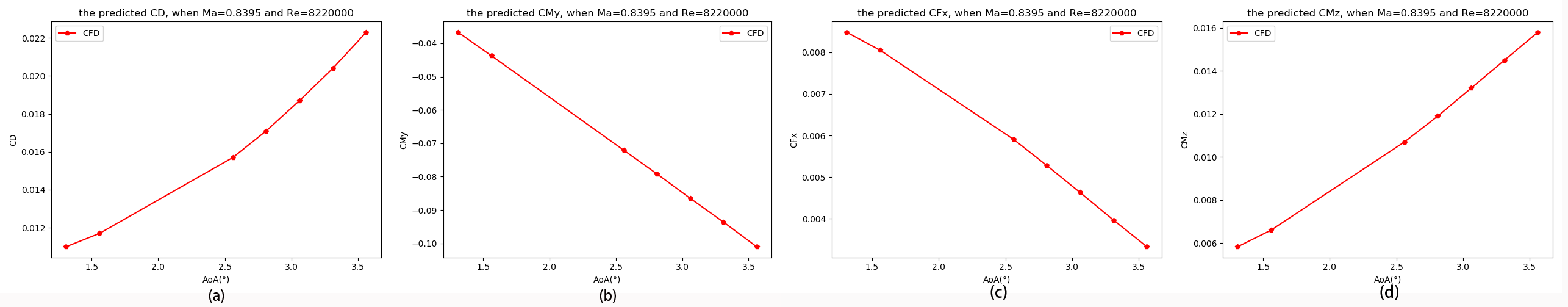}
\caption{The visualization results of coefficient of drag $C_{D}$ (a), coefficient of pitching moment $C_{M_{y}}$ (b), X component coefficient of friction $C_{F_{x}}$ (c) and coefficient of yawing moment $C_{M_{z}}$ (d) with $AoA(\alpha)$  under the condition of $Ma=0.8$ and $Re=8220000$. }\label{fig_visuliazation_of_M6_dataset}
\centering
\end{figure*}

\subsubsection{ONERA M6 Dataset}

The ONERA M6 dataset is given to describe the aerodynamic characteristics of a fixed ONERA M6 airfoil under different flow states \cite{balan2020verification}. We designed 770,000 grids (Fig.\ref{fig_M6_wing_and_grid}) on the surface and the far field of the ONERA M6 airfoil to calculate 755 sample points using SU2. The format and the variation range of design parameters in this dataset are shown in Tab.\ref{tab_M6_input_data_format}. The format of response parameters is shown in Tab.\ref{tab_M6_output_data_format}. We can see that the $AoA(\alpha)$ varies from 1.06 to 4.81, but there are only four different values, so the ONERA M6 dataset is a limited, non-linear and sparse dataset. Fig.\ref{fig_visuliazation_of_M6_dataset} illustrates the variation of coefficient of drag $C_{D}$, coefficient of pitching moment $C_{M_{y}}$, X component coefficient of friction $C_{F_{x}}$ and coefficient of yawing moment $C_{M_{z}}$ with $AoA(\alpha)$ when $Ma=0.8$ and $Re=8220000$.

\begin{figure}[tbh]
\centering
\includegraphics[scale=0.7]{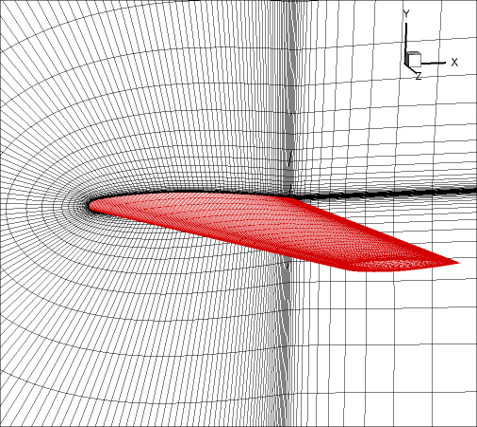}
\caption{The M6 wing and the grids on the surface and far field of the swing.}\label{fig_M6_wing_and_grid}
\centering
\end{figure}

\begin{table}[htb]
\centering
\begin{tabular}{cccc}
\hline
variables & $Ma$ & $AoA(\alpha)$ & $Re$\tabularnewline
\hline
significance & mach number & angle of attack & reynolds number\tabularnewline
start & 0.8045 & 1.06 & 8030000\tabularnewline
end & 0.8395 & 4.81 & 13000000\tabularnewline
step & 0.005 & 1.25 & 250000\tabularnewline
\hline
\end{tabular}
\caption{The data format and variation range of design parameters in M6 test
case.}\label{tab_M6_input_data_format}
\centering
\end{table}

\begin{table*}[tbh]
\centering
\begin{tabular}{cccc}
\hline
variables & $C_{L}$ & $C_{D}$ & $C_{F}$\tabularnewline
significance & coefficient of lift & coefficient of drag & coefficient of friction\tabularnewline
\hline
variables & $C_{M_{x}}$ & $C_{M_{y}}$ & $C_{M_{z}}$\tabularnewline
significance & coefficient of rolling moment  & coefficient of pitching moment & coefficient of yawing moment \tabularnewline
\hline
variables & $C_{F_{x}}$ & $C_{F_{y}}$ & $C_{F_{z}}$\tabularnewline
significance & X component coefficient of friction & Y component coefficient of friction & Z component coefficient of friction\tabularnewline
\hline
\end{tabular}
\caption{The data format of response parameters in ONERA M6 dataset.}\label{tab_M6_output_data_format}
\centering
\end{table*}

\subsection{Experimental Results}

We compare 5 different neural networks, namely FCN \cite{tenney2018deep}, ClusterNet \cite{white2020fast}, GAN/cGAN \cite{farimani2017deep}, RBF-GAN and RBFC-GAN, among which the last three are generative models we pay special attention to. The learning rate is 0.0001, with the batch size 128, and the number of iterations 2000 within all models. The G of the generative models takes a 62 dimensional noise which follows a uniform distribution $U(0,1)$ as a input. The initializations of centers and width of the hidden neurons in RBF-GAN or RBFC-GAN follow a uniform distribution $U(0,1)$ and a normal distribution $N(0.5, 0.2)$, respectively. The five models compared in this paper are implemented based on tensorflow framework. All programs run on four Tesla K80 GPUs. The errors of these models are evaluated by MSE and MSPE:

\begin{equation}
\nonumber
\label{equ_MSE_and_MSPE}
\begin{cases}
  \ensuremath{MSE=\frac{1}{n}\sum_{i=1}^{n}\left(f\left(x_{i}\right)-y_{i}\right)^{2}}\\
  \ensuremath{MSPE=\frac{1}{n}\sum_{i=1}^{n}\frac{\left\Vert f\left(x_{i}\right)-y_{i}\right\Vert ^{2}}{\left\Vert y_{i}\right\Vert ^{2}}\times100\%}
\end{cases}
\end{equation}
where $y_{i}$ denotes the real response values, $f(x_{i})$ denotes the generated response values under the same design parameters.

\subsubsection{Burgers' Dataset}
\begin{figure*}[htbp]
\centering
\includegraphics[scale=0.3]{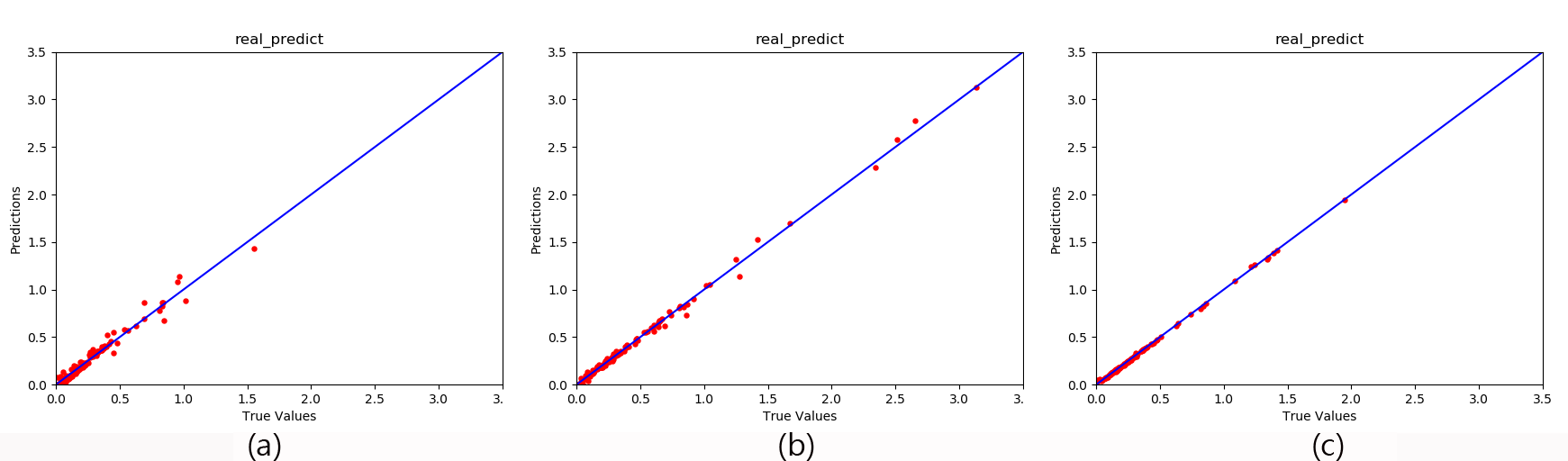}
\caption{The velocity $u$ generated by GAN (a), RBF-GAN(b) and RBFC-GAN (c)
in Burgers' dataset.}\label{fig_burgers_line}
\centering
\end{figure*}

\begin{figure*}[htbp]
\centering
\includegraphics[scale=0.22]{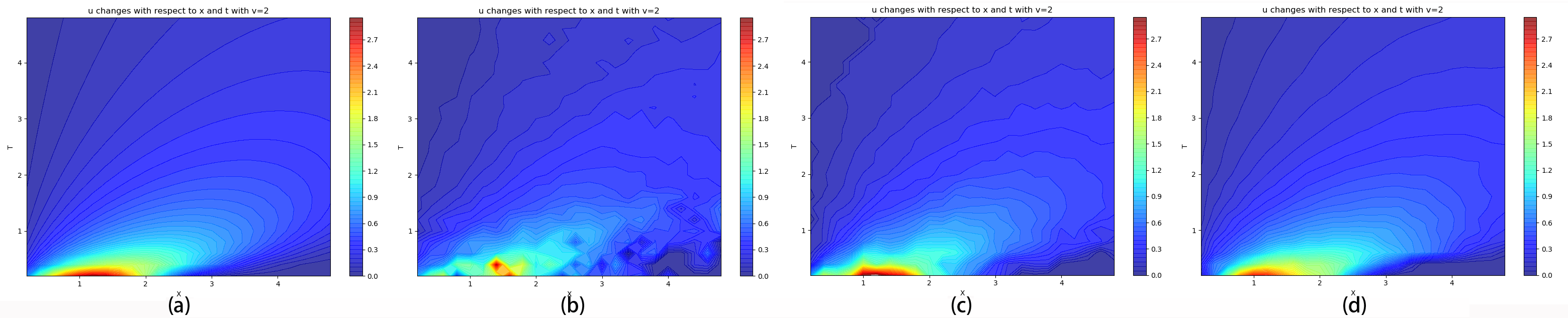}
\caption{The visualization results of velocity $u$ generated by CFD (a), classic GAN (b), RBF-GAN
(c) and RBFC-GAN(d).}\label{fig_burgers_compare_visualization}
\centering
\end{figure*}

The Burgers' equation has numerical solutions, hence the model errors is evaluated using the Burgers' equation itself instead of validation sets or test sets. (i.e. there is no need to divide the whole datasets into training sets, validation sets and test sets). Tab.\ref{tab_burgers_error} shows the comparison results of the above models with the Burgers' dataset. We can learn that the RBF-GAN is more accurate than GAN, but worse than the FCN and the ClusterNet. Notably, the MSE and the MSPE of the RBFC-GAN are the lowest among these models.

\begin{table}[htb]
\centering
\begin{tabular}{cccc}
\hline
Method     & Structure                                         & MSE                 & MSPE\tabularnewline
\hline
FCN        & 5,32                                              & $5.26\times10^{-4}$ & 5.92\%\tabularnewline
ClusterNet & 4;1,1; 4,5                                        & $4.31\times10^{-4}$ & 5.36\%\tabularnewline
GAN        & \tabincell{c}{G(62,64{*}1,4)\\D(4,64{*}1,1)}      & $2.07\times10^{-3}$ & 299.19\%\tabularnewline
RBF-GAN    & \tabincell{c}{G(62,1024{*}1,4)\\ D(4,1024{*}1,1)} & $9.87\times10^{-4}$ & 10.86\%\tabularnewline
RBFC-GAN   & \tabincell{c}{G(62,128{*}2,4)\\ D(4,(42,43,43),1)}& $8.21\times10^{-5}$ & 2.96\%\tabularnewline
\hline
\end{tabular}
\caption{The test MSE and MSPE of the FCN, the ClusterNet, the GAN, the RBF-GAN and the RBFC-GAN based on Burgers' dataset. In the column related to structure, ``5,32'' means that the optimal FCN has 5 hidden layers, each hidden layer has 32 neurons. ``4;1,1; 4,5'' means that the optimal ClusterNet consists of 4 clusters, each of which consists of a context network and a functional network. The context network has 1 hidden layer, each of which has only one neuron. The functional network has 4 hidden layers, each of which has 5 nodes. ``G(62,128{*}2,4)'' indicates that the FCN-based generator has four layers, the number of neurons in the
input layer, and the output layer is 62 and 4, respectively. ``128{*}2'' denotes that FCN-based generator has 2 hidden layers, each of which has 128 neurons. ``D(4,(42,43,43),1)'' indicates that the RBFs cluster-based discriminator has three layers, 4 and 1 denotes the number of neurons in the input layer, and the output layer, respectively. 42,43 and 43 denotes the number of hidden neurons of the Guassian kernel-based RBFNN, the Laplace kernel-based RBFNN and the inverse multiquadrics kernel-based RBFNN, respectively.}\label{tab_burgers_error}
\centering
\end{table}

To compare the ability of the three generative models to generate Burgers' data, the visualization results of velocity $u$ generated by the GAN, the RBF-GAN and the RBFC-GAN in Tab.\ref{tab_burgers_error} are shown in Fig.\ref{fig_burgers_line} and \ref{fig_burgers_compare_visualization}. From Fig.\ref{fig_burgers_line}, we can learn that the velocity $u$ generated by GAN satisfies $u<2.0$, which implies that the GAN can only generate low-velocity data, i.e. model collapse \cite{bang2018mggan}. Compared with the GAN, the RBF-GAN can generate data with $u>2.0$, which means that the RBF-GAN alleviates the problem of mode collapse. The RBFC-GAN is more accurate than the RBF-GAN, but still faces the problem of model collapse. Fig.\ref{fig_burgers_compare_visualization} also reflects the same results. The red part in the lower left corner of subgraph (c) is close to (a), which means that the RBF-GAN has the ability to generate high-velocity Burgers' data. Besides, compared with subgraph (b) and (c) in Fig.\ref{fig_burgers_compare_visualization}, the contour lines in subgraph (d) show a smooth change, which is close to subgraph (a).

\subsubsection{Cylindrical Laminar dataset}

\begin{figure*}[htbp]
\centering
\includegraphics[scale=0.35]{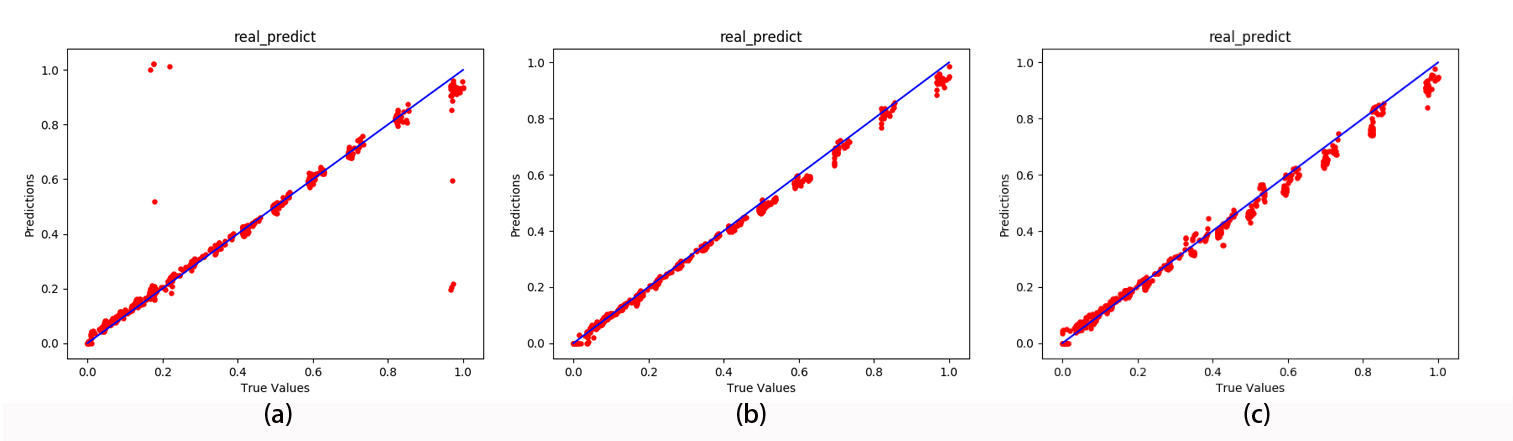}
\caption{The pressure $P$ generated by cGAN (a), RBF-cGAN (b) and RBFC-cGAN (c). }\label{fig_NS_P_line}
\centering
\end{figure*}

\begin{figure*}[htbp]
\centering
\includegraphics[scale=0.35]{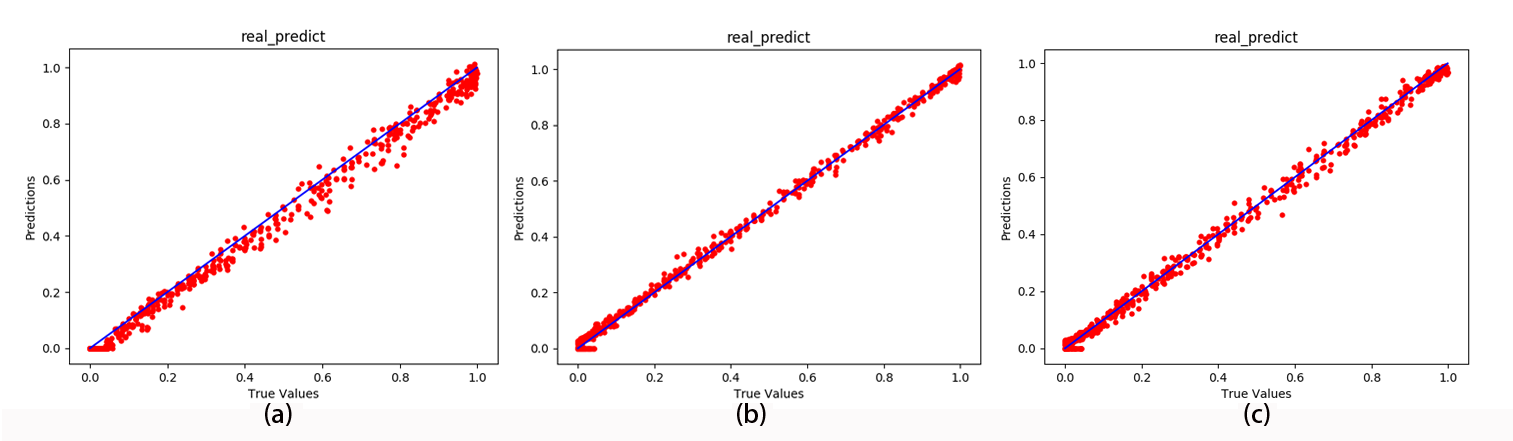}
\caption{The friction in the x-coordinate $F_{x}$ generated data by cGAN (a), RBF-cGAN (b) and RBFC-cGAN (c).}\label{fig_NS_Fx_line}
\centering
\end{figure*}

\begin{figure*}[htbp]
\centering
\includegraphics[scale=0.15]{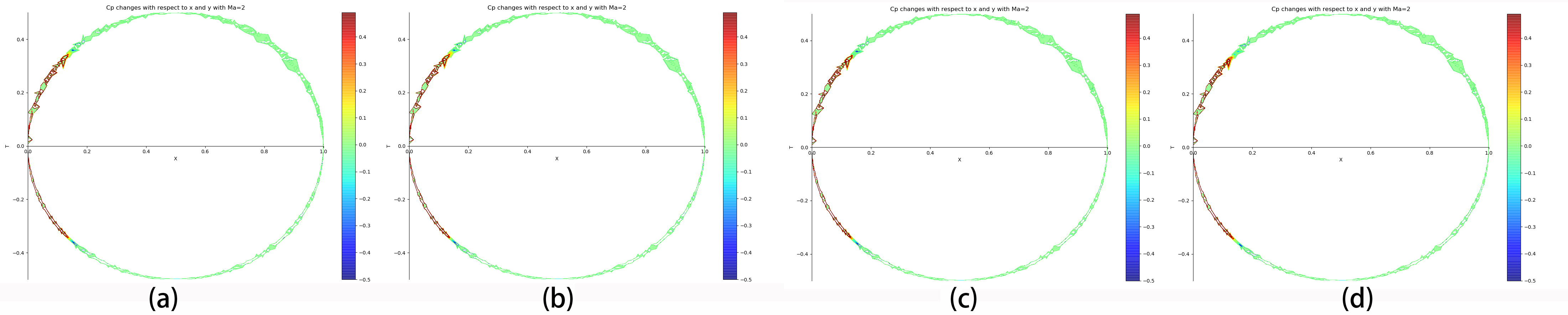}
\caption{The visualization results of coefficient of pressure $Cp$ generated by CFD (a), classic GAN (b), RBF-GAN (c) and RBFC-GAN(d). }\label{fig_NS_compare_visualization}
\centering
\end{figure*}
In the experiment of cylindrical laminar, we test the same neural networks as in Burgers' experiment. Unlike Burgers' dataset, the Navier-Stokes equation is more complex, and it is not realistic to calculate MSE and MSPE through this equation. Therefore, we divide the  whole data set into three sets: a training set, a validation set and a test set with a ratio of 8:1:1. Considering the inability of GANs to generate data under specific flow states, we adopted cGANs to simplify the calculation of MSE and MSPE. Tab.\ref{tab_NS_result} shows the test errors of the five neural networks based on the cylindrical laminar dataset. We can learn that the RBF-cGAN and the RBFC-cGAN are more accurate than the cGAN, and similar with the FCN and the ClusterNet.

\begin{table}[tbh]
\centering
\begin{tabular}{cccc}
\hline
model      & Structure                                          & MSE                 & MSPE\tabularnewline
\hline
FCN        & 5,64                                               & $2.48\times10^{-4}$ & 1.72\%\tabularnewline
ClusterNet & 4;1,5;3,64                                         & $6.49\times10^{-4}$ & 4.73\%\tabularnewline
cGAN       & \tabincell{c}{G(62,512{*}1,4)\\ D(7,512{*}1,1)}    & $3.98\times10^{-3}$ & 6.40\%\tabularnewline
RBF-cGAN   & \tabincell{c}{G(62,512{*}2,4)\\ D(7,512{*}1,1)}    & $2.60\times10^{-4}$ & 2.05\%\tabularnewline
RBFC-cGAN  & \tabincell{c}{G(62,128{*}3,4)\\ D(7,(42,43,43),1)} & $4.06\times10^{-4}$ & 4.42\%\tabularnewline
\hline
\end{tabular}
\caption{The test MSE and MSPE of the FCN, ClusterNet, cGAN, RBF-cGAN and
RBFC-cGAN based on Cylindrical Laminar dataset. }\label{tab_NS_result}
\centering
\end{table}

Fig.\ref{fig_NS_P_line} and \ref{fig_NS_Fx_line} shows the pressure $P$ and the friction in x-coordinate $F_{x}$ generated by the cGAN, the RBF-cGAN and the RBFC-cGAN in Tab.\ref{tab_NS_result}, respectively. We can learn that the $P$ and $F_{x}$ generated by cGAN is less accurate than that of the RBF-cGAN and the RBFC-cGAN. Fig.\ref{fig_NS_compare_visualization} shows the visualization results of coefficient of pressure $C_{p}$ by CFD, the cGAN, the RBF-cGAN and the RBFC-cGAN. The four subgraphs in Fig.\ref{fig_NS_compare_visualization} are similar, which indicates that the cGAN, the RBF-cGAN and the RBFC-cGAN can learn the distribution of pressure coefficient of cylindrical laminar dataset.

\subsubsection{ONERA M6 Dataset}

Similar with cylindrical laminar dataset, the whole ONERA M6 dataset are divided into a training set, a validation set and a test set with a ratio of 8:1:1. Besides, cGANs are still adopted in this dataset. In this experiment, the accuracy, stability and computational cost of the models are mainly discussed.

In terms of accuracy, we still choose MSE and MSPE as the indicators to evaluate the accuracy of models. Tab.\ref{tab_M6_result} shows the average errors of the ONERA M6 data generated by different models. We can learn that the RBF-cGAN and the RBFC-cGAN are more accurate than the cGAN, and similar to the FCN and the ClusterNet. Besides, the RBFC-cGAN is more accurate than the RBF-cGAN and the ClusterNet. Fig.\ref{fig_M6_CD_line} and \ref{fig_M6_CMY_line} show the $C_{D}$ and $C_{M_{y}}$ generated by three generative models. We can see that the $C_{D}$ and $C_{M_{y}}$ generated by the cGAN do not closely surround the 45-degree diagonal, but the data generated by the other two models are tightly clustered around the 45-degree diagonal. The visualization results of the remaining response parameters are similar. Fig.\ref{fig_M6_compare} illustrates the variation of $C_{D}$, $C_{M_{y}}$, $C_{F_{x}}$ and $C_{M_{z}}$ with $AoA$ respectively under the condition of $Ma=0.8395$ and $Re=8220000$. Compared with CFD, the data generated by cGAN always have large errors. As a result, the data generated by the RBF-cGAN and the RBFC-cGAN are more accurate than that generated by the cGAN, and the RBFc-GAN is more accurate than RBF-cGAN.

\begin{table}[tbh]
\centering
\begin{tabular}{cccc}
\hline
Method     & Structure                                      & MSE                & MSPE\tabularnewline
\hline
FCN        & 5 10                                           & $2.07\times10^{-6}$ & 0.23\%\tabularnewline
ClusterNet & 4; 1 16; 3 32                                  & $8.29\times10^{-6}$ & 0.36\%\tabularnewline
cGAN       & \tabincell{c}{G(62,512,9)\\D(12,512,1)}        & $5.48\times10^{-5}$ & 122.25\%\tabularnewline
RBF-cGAN   & \tabincell{c}{G(62,512,9)\\D(12,512,1)}        & $1.90\times10^{-5}$ & 0.72\%\tabularnewline
RBFC-cGAN  & \tabincell{c}{G(62,128,9)\\D(12,(42,43,43),1)} & $1.24\times10^{-5}$ & 0.28\%\tabularnewline
\hline
\end{tabular}
\caption{The test MSE and MSPE of classic cGAN and RBF-cGAN with M6 dataset. }\label{tab_M6_result}
\centering
\end{table}

\begin{figure*}[tbh]
\centering
\includegraphics[scale=0.36]{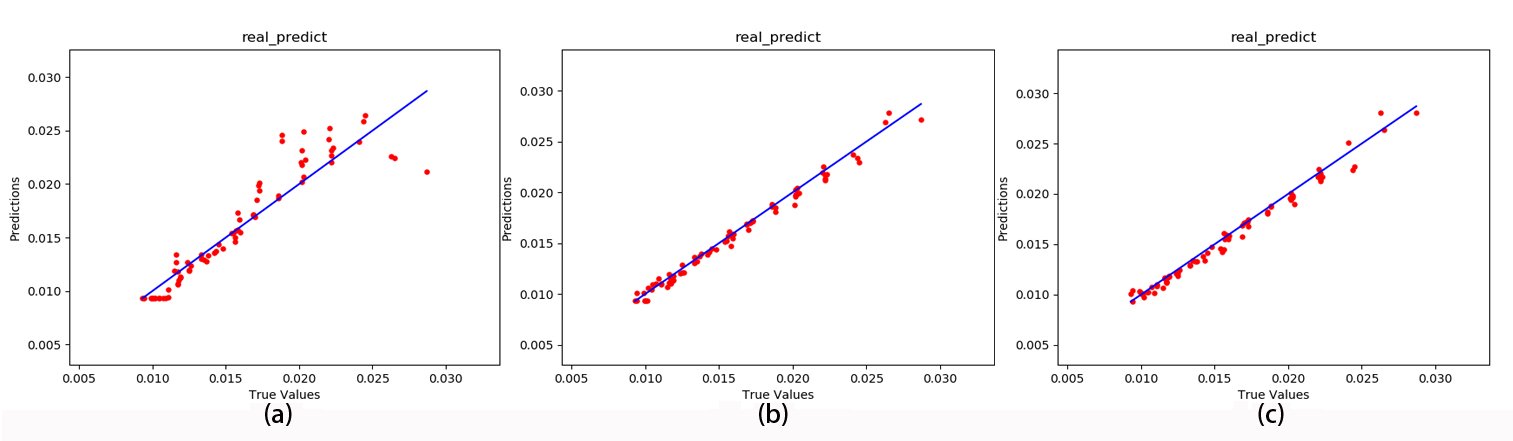}
\caption{The $C_{D}$ generated by cGAN(a), RBF-cGAN(b) and RBFC-cGAN(c). }\label{fig_M6_CD_line}
\centering
\end{figure*}

\begin{figure*}[tbh]
\centering
\includegraphics[scale=0.35]{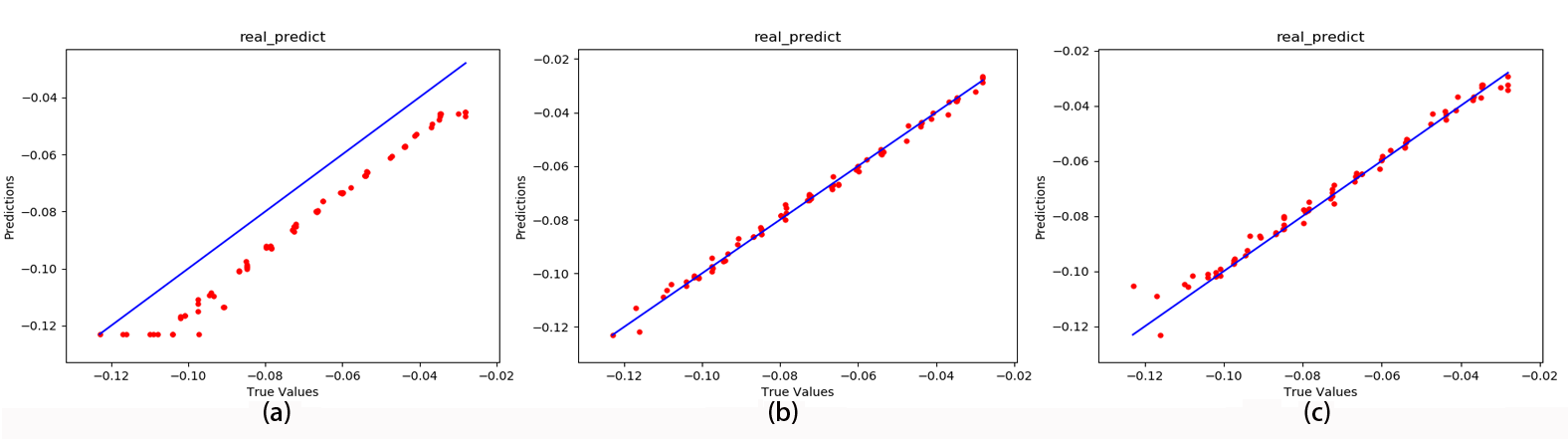}
\caption{The $C_{M_{y}}$ generated by cGAN(a), RBF-cGAN(b) and RBFC-cGAN(c). }\label{fig_M6_CMY_line}
\centering
\end{figure*}

\begin{figure*}[tbh]
\centering
\includegraphics[scale=0.25]{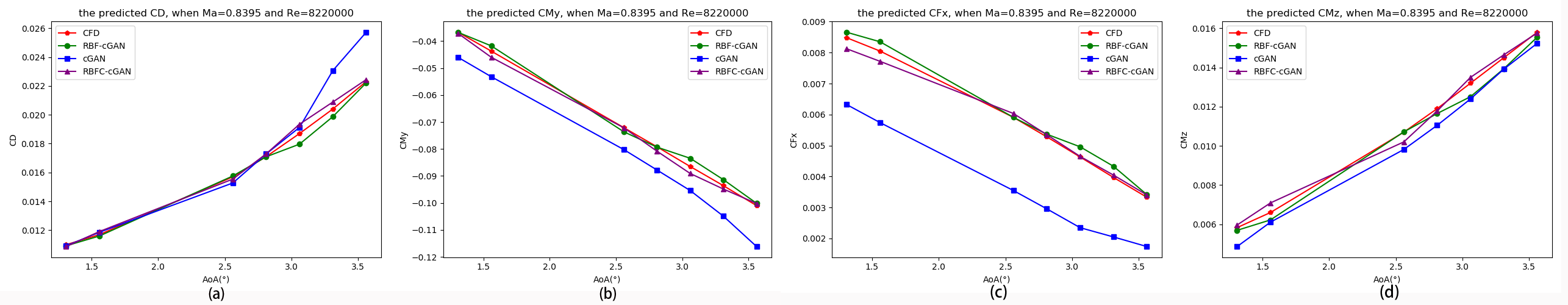}
\caption{The variation of drag coefficient $C_{D}$ (a), coefficient of pitching moment $C_{M_{y}}$ (b), X component coefficient of friction $C_{F_{x}}$ (c) and coefficient of yawing moment $C_{M_{z}}$ (d) with $AoA$ under the condition of $Ma=0.8395$ and $Re = 8220000$.}\label{fig_M6_compare}
\centering
\end{figure*}

In terms of stability, the validation errors (MSE) of the three generative models during the training process are compared in Fig.\ref{fig_stability}. We can learn that when these three models converged (epoch > 750), the validation MSE of the RBF-cGAN and the RBFC-cGAN are lower than that of the cGAN, and the validation error of the RBFC-cGAN are lower than that of the RBF-cGAN. Besides, the RBF-cGAN converges faster than other models. According to the variation of the validation MSE of the three models, the entire iteration period is divided into three intervals, i.e. {[}0,125), {[}125,750) and {[}750,2000{]}. Tab.\ref{tab_stability} illustrates the standard deviation of the validation MSE of the three models. The stability improvement $\eta$ of a model after convergence is calculated using
\begin{equation}
\nonumber
\label{equ_stability_error}
\eta=\frac{||\hat{\sigma}-\sigma||}{||\sigma||}\times100\%
\end{equation}
where $\hat{\sigma}$ and $\sigma$ denote standard deviations of two different models. Compared with cGAN, the stability improvement $\eta_{1}$ of RBF-cGAN and $\eta_{2}$ of RBFC-cGAN are 34.62\% and 72.31\%, respectively. However the stability improvement of RBFC-cGANs is based on the cost of convergence time.

As for the computational costs, because SU2 cannot run on GPU, we run all approaches on CPU to compare the average computational cost for generating a ONERA M6 data, which is shown in Tab.\ref{tab_M6_average_time}. In term of average time, SU2 takes the longest time to calculate a data, while the rest of models can output a data within 0.05 second, among which the FCN takes the shortest time to output a data, while the cGAN takes longest time to do so. Both of the RBF-cGAN and RBFC-cGAN are faster than the cGAN, but slower than the FCN. In terms of hardware resources, SU2 occupies the most CPU and memory usage, followed by the clusterNet, the RBF-cGAN. The RBFC-cGAN, the cGAN and the FCN occupy the least hardware resources.

\begin{figure}[bht]
\centering
\includegraphics[scale=0.35]{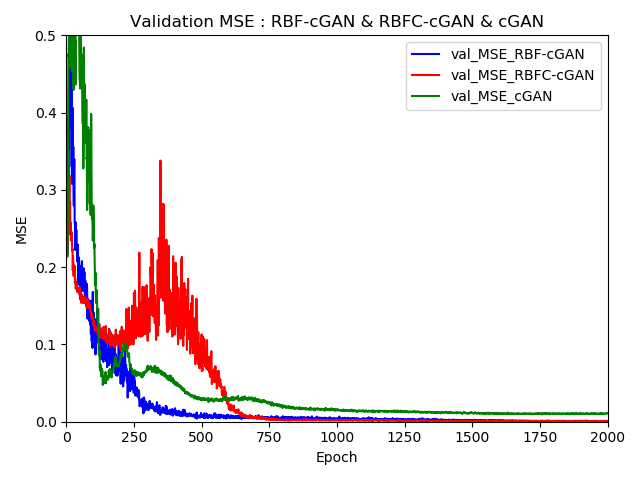}
\caption{The validation MSE during the training process of the three generative models.}\label{fig_stability}
\centering
\end{figure}

\begin{table}[tbh]
\centering
\begin{tabular}{r@{\extracolsep{0pt}.}lr@{\extracolsep{0pt}.}lr@{\extracolsep{0pt}.}lr@{\extracolsep{0pt}.}l|r@{\extracolsep{0pt}.}l}
\hline
\multicolumn{2}{c}{models} & \multicolumn{2}{c}{{[}0,125)} & \multicolumn{2}{c}{{[}125,750)} & \multicolumn{2}{c|}{{[}750,2000{]}} & \multicolumn{2}{c}{$\eta$ }\tabularnewline
\hline
\multicolumn{2}{c}{cGAN} & 0&16 & 0&027 & 0&0013 & \multicolumn{2}{c}{\textendash{}}\tabularnewline
\multicolumn{2}{c}{RBF-cGAN} & 0&082 & 0&048 & 0&00085 & 34&62\%\tabularnewline
\multicolumn{2}{c}{RBFC-cGAN} & 0&082 & 0&074 & 0&00036 & 72&31\%\tabularnewline
\hline
\end{tabular}
\caption{The standard deviation of the validation error of the cGAN, RBF-cGAN and RBFC-cGAN in three intervals divided according to the training epochs. $\eta$ calculated by (\ref{equ_stability_error}) denotes the stability improvement compared with cGAN when the model converges.}\label{tab_stability}
\centering
\end{table}

\begin{table}[tbh]
\centering
\begin{tabular}{cccc}
\hline
Method      & Average time(s) & CPU usage(\%) & Memory usage(MB)\tabularnewline
\hline
SU2         & 6600            &89             &15194.33\tabularnewline
FCN         & 0.0048          &41.56          &174.23\tabularnewline
ClusterNet  & 0.0085          &86.32          &212.45\tabularnewline
cGAN        & 0.024           &53.00          &176.82\tabularnewline
RBF-cGAN    & 0.0206          &76.21          &186.98\tabularnewline
RBFC-cGAN   & 0.0065          &58.79          &188.58\tabularnewline
\hline
\end{tabular}
\caption{The average computational cost of approaches to output a data with ONERA M6 dataset. }\label{tab_M6_average_time}
\centering
\end{table}

\subsection{Results Analysis}

\subsubsection{The Errors}
Consider a simple Gaussian kernel function $K(x,v)$ with $\sigma=1$:

\begin{equation}
\label{equ_proof_of_dimension}
\begin{aligned}
K(x,v)&=exp(-(x-v)^{2})\\&=exp(-x^{2})\cdot exp(-v^{2})\cdot exp(2xv)
\end{aligned}
\end{equation}
where $x$ and $v$ denote the input and the center of $K(x,v)$, respectively. $exp(2xv)$ can be expanded using the Taylor expansion:

\begin{equation}
\label{equ_proof_of_dimension_taylor}
\begin{aligned}
K(x,v)&=exp(-x^{2})\cdot exp(-v^{2})\cdot \sum_{i=0}^{\infty}\frac{(2xv)^{i}}{i!}
\\&=\sum_{i=0}^{\infty}exp(-x^{2})\cdot exp(-v^{2})\sqrt{\frac{2^{i}}{i!}}x^{i}\sqrt{\frac{2^{i}}{i!}}v^{i}
\end{aligned}
\end{equation}

Let $\Phi(x)=exp(-x^{2})\cdot(1,\sqrt{\frac{2}{1}}x,\sqrt{\frac{2^{2}}{2!}}x^{2},\cdots)^T$, then $K(x,v)=\Phi(x) \cdot \Phi(v)$. As we can see, $\Phi(x)$ is an infinite dimensional vector, which demonstrated that the Gaussian kernel function maps low dimensional data into a higher dimensional space. The distance between two similar distributions in the low dimensional space (i.e. the input space) is magnified in a higher dimensional space (i.e. the hidden space). Hence, an RBF-D improves the ability to distinguish the real data from the generated data, which directly improves the quality of data generated by the corresponding generator.

In the RBFC-GAN, three different RBFs are linearly combined into a more complex multi-kernel RBF \cite{kou2017multi}. A multi-kernel RBF means that the low dimensional inputs can be mapped into a higher dimensional space than an RBF-GAN, which can further reduce the prediction error \cite{white2020fast} under the condition of the same number of neurons.

\subsubsection{The Stability}
We still consider the $K(x,v)$ with $\sigma=1$:
\begin{equation}
\label{equ_local_updating}
K(x,v)=exp(-(x-v)^{2})=\begin{cases}
1 & x=v\\
0 & |x-v| \gg 0
\end{cases}
\end{equation}

For any given input $x$, only hidden neurons whose center $v$ closes to $x$ are activated, which causes a small variation in $D(x)$. However, the hidden neurons of a FCN are all activated, which causes a relatively large variation in $D(x)$. Therefore, the RBF-GAN is more stable than FCN-based GANs.

Compared with the RBF-GAN, the form of clusters of RBFNNs can further improve the stability of the RBFC-GAN under the condition of the same number of neurons. In the RBFC-GAN, different RBFs are employed simultaneously to overcome the fluctuations of gradients. The gradients of weights in an RBFC-GAN are as flollows:

\begin{equation}\label{equ_derivation_of_wzj}
\frac{\partial L_{D}}{\partial W z j}=\frac{\partial L_{D}}{\partial D(x)} \cdot \frac{\partial D(x)}{\partial W z j}=\frac{\partial L_{D}}{\partial D(x)} \cdot \lambda_{z} \cdot g_{z}\left(x_{i}, v_{j}, \sigma_{j}\right)
\end{equation}
where $L_{D}$ (\ref{equ_loss_D_norm}) denotes the loss function of the RBFC-D, $W_{zj}$ denotes the weights of the RBFC-D. The gradient of weights in the RBFC-GAN does not fluctuates radically because the weighting of three different radial basis functions makes the value of $D(x)$ and $L_{D}$ more stable after the RBFC-GAN converged.

\subsubsection{The computational cost}
In term of average time for generating a M6 data, the local updating strategy of RBFNNs accelerates the convergence procedure, which results in that both the RBF-GAN and RBFC-GAN are faster than the FCN-based GAN. Besides, the clusters of RBFc-GAN do not interfere with each other, and the scale of each cluster is relatively small, hence, the RBFC-GAN is faster than RBF-GAN. In term of hardware resources uages, the matrix multiplication in GANs is replaced by RBF, which is a more complex operator. Therefore, the hardware resources usage of RBF-GAN and RBFC-GAN are more than FCN and FCN-based GANs.

\section{Conclusions}

In this paper, the nonlinear full domain FFD generation and regression based on GANs is proposed. We firstly prove that an RBFNN is the optimal discriminator of a GAN while dealing with nonlinear sparse data. Then, an RBF-GAN and an RBFC-GAN are further proposed. Compared with existing models, the RBF-GAN and the RBFC-GAN not only reduce the MSE and MSPE of generated FFD, but also improve the stability of GANs. Finally, three different datasets are used to validate the feasibility of our models. The detailed conclusions are as follows:

\begin{enumerate}
\item[a)] we prove that an RBFNN is the optimal discriminator of a GAN while dealing with nonlinear sparse data;
\item[b)] the RBF-GAN is more accurate and stable than GANs/cGANs in the field of flow field reconstructions;
\item[c)] the RBFC-GAN is more accurate and stable than GANs/cGANs and RBF-GANs in the same field;
\item[d)] compared with cGAN, the stability of RBF-cGAN and RBFC-cGAN improves by 34.62\% and 72.31\%, respectively.
\end{enumerate}

In addition, we firmly believe that the proposed RBF-GAN and RBFC-GAN are not only suitable for flow field reconstructions, but also available for many other fields (e.g. the predictions of aerodynamic performance and optimization designs, etc) in the area of aerodynamics. Besides, the optimal generator of GANs is still unsolved, which will be one of our future works.

\section*{Acknowlegment}

This work was supported by the National Numerical Wind Tunnel Project No.13RH19ZT6B1 and No.13RH19ZT62B.

Thanks to Professor. Zhonghua Han and Chenzhou Xu from School of Aeronautics, Northwestern Polytechnical University for providing us with technical support about SU2.

\bibliographystyle{IEEEtran}
\bibliography{Flow_Field_Reconstructions_with_GANs_based_on_Radial_Basis_Functions}

\begin{IEEEbiography}[{\includegraphics[width=1in,height=1.25in,clip,keepaspectratio]{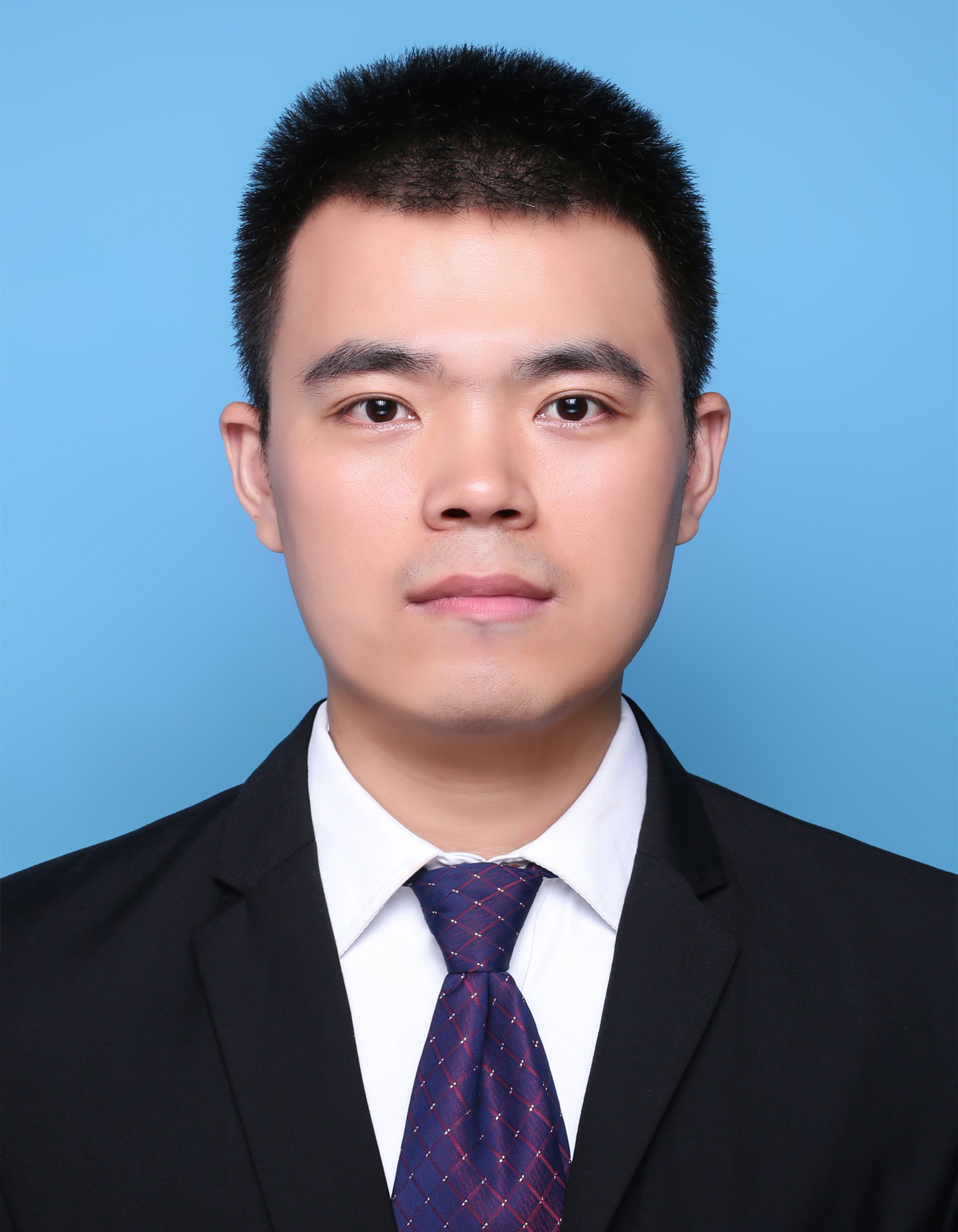}}]{Liwei Hu} received his B.S. degree in software engineering from the School of Software, Hebei Normal University, Shijiazhuang, Hebei, China, in 2014 as well as an M.S. degree in computer technology from the University of Electronic Science and Technology of China (UESTC), Chengdu, Sichuan, China, in 2018. He is currently pursuing a Ph.D. degree in computer science and technology from the School of Computer Science and Engineering, UESTC. His research fields include aerodynamic data modeling, deep learning and pattern recognition.
\end{IEEEbiography}

\begin{IEEEbiography}[{\includegraphics[width=1in,height=1.25in,clip,keepaspectratio]{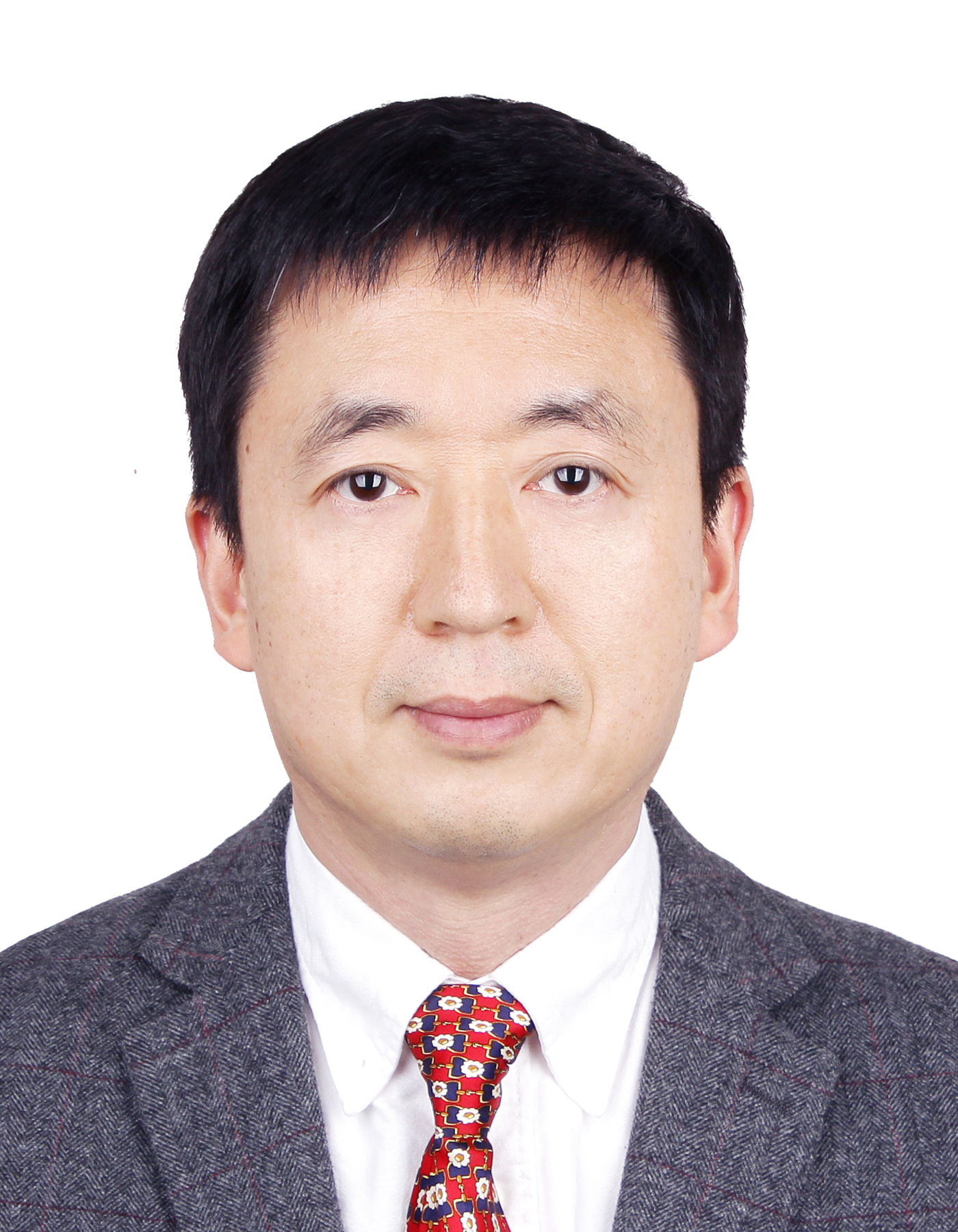}}]{Wenyong Wang} received his B.S. degree in computer science from BeiHang University, Beijing, China, in 1988 and M.S. and Ph.D. degrees from the University of Electronic Science and Technology (UESTC), Chengdu, China, in 1991 and 2011, respectively. He has been a professor in computer science and engineering at UESTC since 2006. Now he is also a special-term professor at Macau University of Science and Technology, a senior member of the Chinese Computer Federation, a member of the expert board of the China Education and Research Network (CERNET) and China Next Generation Internet. His main research interests include next generation Internet, software-defined networks, and software engineering.
\end{IEEEbiography}

\begin{IEEEbiography}[{\includegraphics[width=1in,height=1.25in,clip,keepaspectratio]{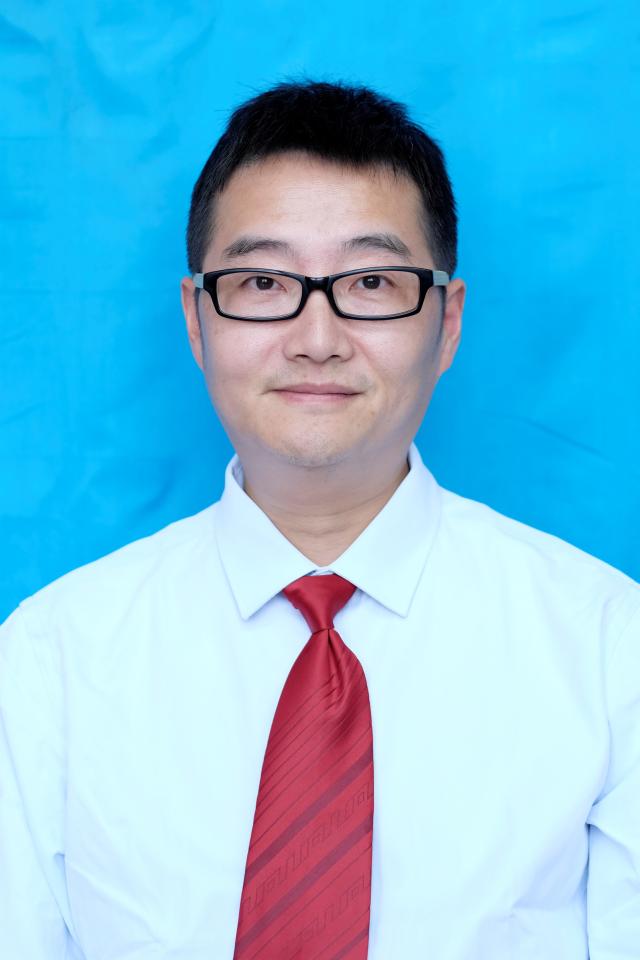}}]{Yu Xiang} received his B.S, M.S. and Ph.D. degrees from the University of Electronic Science and Technology of China (UESTC), Chengdu, Sichuan, China, in 1995, 1998 and 2003, respectively. He joined the UESTC in 2003 and became associate professor in 2006. From 2014-2015, he was a visiting scholar at the University of Melbourne, Australia. His current research interests include computer networks, intelligent transportation systems and deep learning.
\end{IEEEbiography}

\begin{IEEEbiography}[{\includegraphics[width=1in,height=1.25in,clip,keepaspectratio]{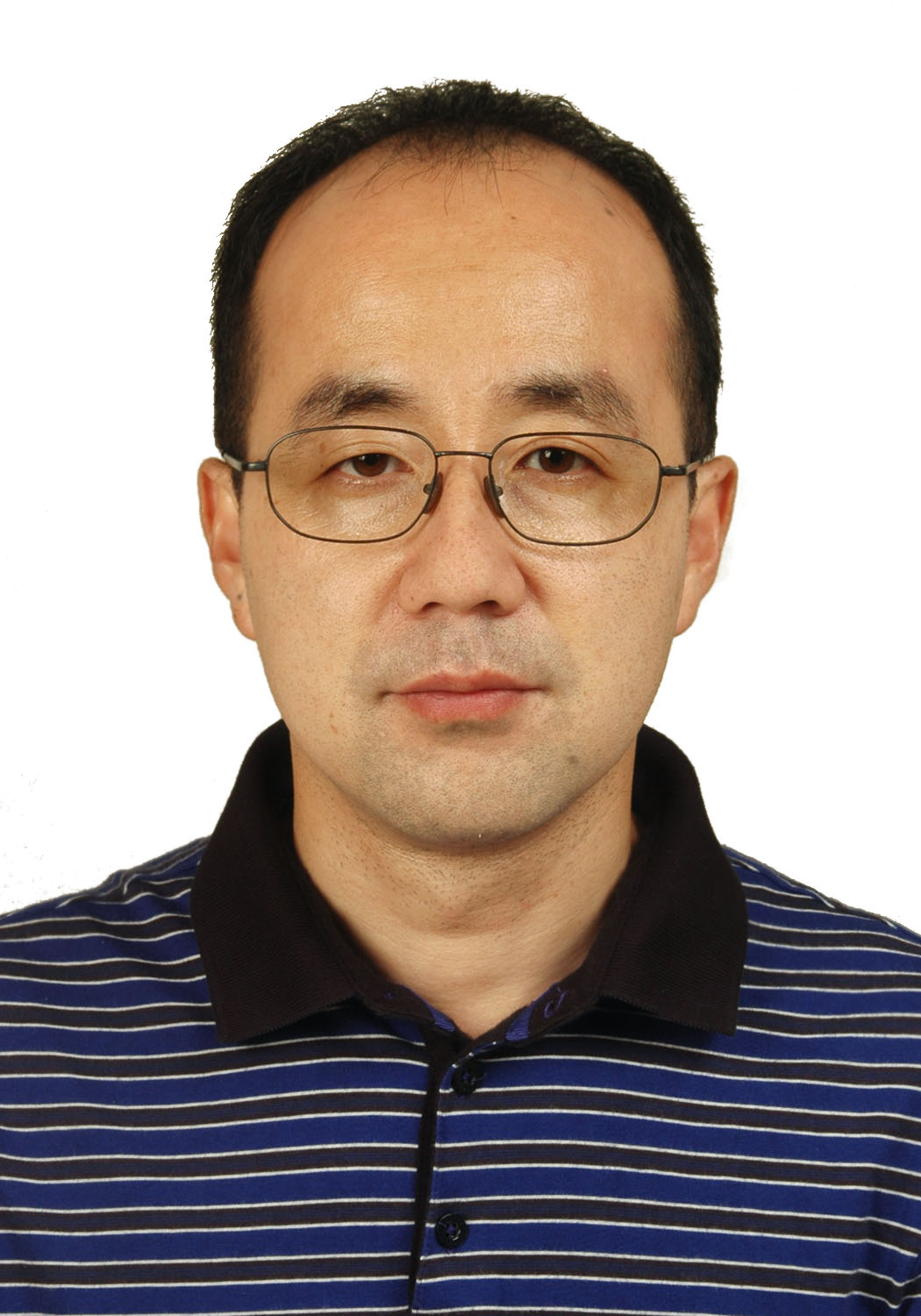}}]{Jun Zhang} received his B.S. and M.S. degrees in electronic engineering from the University of Electronic Science and Technology of China (UESTC) in 1995 and 1998, respectively. From 1998 to 2008, he worked as a senior researcher and engineer in CERNET. He is currently a lecturer at the School of Computer Science and Engineering, UESTC. His current research interests include software-defined networks, machine learning applied in network traffic engineering, and aerodynamics.
\end{IEEEbiography}

\end{document}